\documentclass[10pt,a4paper]{article}
\usepackage[a4paper, total={6in, 8in}]{geometry}
\makeatother

\usepackage{bbm}
\usepackage{tabu}
\usepackage{amsmath}
\usepackage{amsthm}
\usepackage{breqn}
\usepackage{verbatim}
\usepackage{graphicx}
\usepackage{caption}
\usepackage{subcaption}
\usepackage{esvect}
\usepackage{cite}
\usepackage{amssymb}
\usepackage{footnote}
\usepackage{color}
\usepackage{url}
\usepackage{xr}
\usepackage[makeroom]{cancel}
\usepackage{microtype}
\usepackage{dsfont}
\theoremstyle{plain}
\newtheorem{theorem}{Theorem}[section]
\newtheorem{corollary}[theorem]{Corollary}
\newtheorem{definition}[theorem]{Definition}
\newtheorem{lemma}[theorem]{Lemma}
\newtheorem{proposition}[theorem]{Proposition}
\newtheorem{remark}{Remark}
\newtheorem*{problem statement}{Problem Statement}
\newtheorem{feasibility problem}{Problem}

\newtheorem{assumption}{Assumption}
\definecolor{myred}{RGB}{108,0,0}
\graphicspath{{./Figures/}}

\DeclareMathOperator*{\argmin}{arg\,min}
\newcommand{\real}{\mathbb{R}}
\newcommand{\reg}{\mathcal{R}}
\newcommand{\sreg}{\mathcal{S}}

\newcommand{\algeb}{\mathcal{F}}
\newcommand{\pr}{\mathbb{P}}
\newcommand{\borel}{\mathcal{B}}
\newcommand{\normal}{\mathcal{N}}
\newcommand{\expec}{\mathrm{E}}
\newcommand{\cov}{\mathrm{Cov}}

\newcommand{\sys}{\mathrm{S}}
\newcommand{\proj}{\textbf{proj}}
\newcommand{\uns}{q_{\mathrm{abs}}}
\newcommand{\imc}{\mathrm{imc}}

\newcommand{\Y}{\mathrm{Y}}

\newcommand{\adv}{\boldsymbol{\pi}}
\newcommand{\cum}{\mathrm{cum}}
\newcommand{\mul}{\mathrm{mul}}
\newcommand{\avg}{\mathrm{avg}}

\newcommand{\ind}{\mathds{1}}
\newcommand{\tzeta}{\tilde{\zeta}}
\begin{document}
	\title{Formal Analysis of the Sampling Behaviour of Stochastic Event-Triggered Control}
	\author{Giannis Delimpaltadakis, Luca Laurenti, and Manuel Mazo Jr.
		\thanks{The authors are with the Delft University of Technology, The Netherlands. Emails:\{i.delimpaltadakis, l.laurenti, m.mazo\}@tudelft.nl. This work is partially supported by ERC Starting Grant SENTIENT (755953).}}
	\date{}
	\maketitle
	\begin{abstract} 
		Analyzing Event-Triggered Control's (ETC) sampling behaviour is of paramount importance, as it enables formal assessment of its sampling performance and prediction of its sampling patterns. In this work, we formally analyze the sampling behaviour of stochastic linear periodic ETC (PETC) systems by computing bounds on associated metrics. Specifically, we consider functions over sequences of state measurements and intersampling times that can be expressed as average, multiplicative or cumulative rewards, and introduce their expectations as metrics on PETC's sampling behaviour. We compute bounds on these expectations, by constructing appropriate Interval Markov Chains equipped with suitable reward structures, that abstract stochastic PETC's sampling behaviour. Our results are illustrated on a numerical example, for which we compute bounds on the expected average intersampling time and on the probability of triggering with the maximum possible intersampling time in a finite horizon.
	\end{abstract}
	\section{Introduction}
	In the past two decades, Event-Triggered Control (ETC), followed by its sibling Self-Triggered Control, has constituted the primary research focus of the control systems community towards reducing resource consumption in Networked Control Systems \cite{astrom2002comparison,tabuada2007etc,heemels2012periodic,postoyan2014framework,girard2015dynamicetc,dynamic_stochastic_etc,stochastic_petc,stochastic_etc_delay,on_etc_stochastic_tac2020,demirel2016trade,postoyan_interevent,tallapragada,gabriel_hscc,arman_formal_etc,gleizer2020scalable,delimpaltadakis_cdc_2020,delimpaltadakis2020abstracting}. ETC is a sampling paradigm, where communication between the sensors and the controller takes place only when a state-dependent triggering condition is satisfied. Even though ETC's event-based sampling typically reduces the amount of communications (compared to conventional periodic sampling), it also generates an erratic and generally a-priori unknown \emph{sampling behaviour}. 
	Understanding and predicting ETC's sampling behaviour is of paramount importance as it enables: a) forecasting an ETC design's sampling performance (e.g., how frequently or erratically an ETC system is expected to sample), and b) scheduling data traffic\footnote{Traffic scheduling is much more trivial with periodic sampling, where the generated traffic is known beforehand.} in networks shared by multiple ETC loops, to avoid packet collisions.
	
	Even though analyzing ETC's sampling behaviour is fundamental, research around it is scarce \cite{demirel2016trade,postoyan_interevent,tallapragada,arman_formal_etc,gabriel_hscc,gleizer2020scalable,delimpaltadakis_cdc_2020,delimpaltadakis2020abstracting}. One branch of it composes of analytic approaches \cite{demirel2016trade,postoyan_interevent,tallapragada}. In particular, \cite{demirel2016trade} studies deadbeat stochastic linear PETC (periodic ETC, e.g. \cite{heemels2012periodic}; a practical variant of ETC) systems. Owing to deadbeat control, studying the sampling behaviour of the system simplifies to analyzing a Markov chain, which can be used to compute quantitative metrics over the sampling performance. Nevertheless, assuming deadbeat control is admittedly restrictive. Furthermore, \cite{postoyan_interevent} and \cite{tallapragada} investigate asymptotic properties of the \emph{intersampling times} of 2-D linear ETC systems with quadratic triggering conditions. Despite the interesting insights, these works also suffer from certain limitations: a) they consider only planar systems, b) they are dependent on the type of triggering condition considered, and most importantly c) they do not provide quantitative information on all possible sampling patterns that may be exhibited by an ETC system; as such, they may not be employed for e.g. computing metrics on ETC's sampling performance or predicting its sampling patterns.
	
	The other branch of research studying ETC's sampling behaviour is \emph{abstraction}-based approaches \cite{gabriel_hscc,arman_formal_etc,gleizer2020scalable,delimpaltadakis_cdc_2020,delimpaltadakis2020abstracting}, among which the present work is placed. These construct finite-state systems, abstracting a given ETC system's sampling: the set of their traces contains all possible sequences of intersampling times that may be exhibited by the ETC system. Contrary to the analytic approaches, they do not place restrictive assumptions on the type of controller, and they depend neither on the system dimensions (modulo computational complexity) nor on the triggering condition (except that minor steps in the abstraction's construction might vary). Meanwhile, in exchange for high computational complexity, they provide quantitative information on ETC's diverse sampling patterns, allowing for computing performance metrics and predicting sampling patterns. For example, \cite{arman_formal_etc,gleizer2020scalable,delimpaltadakis_cdc_2020,delimpaltadakis2020abstracting} construct abstractions in the context of ETC traffic scheduling, while \cite{gabriel_hscc} utilizes them to compute the minimum average intersampling time.
	
	Thus far, in \cite{arman_formal_etc,gleizer2020scalable,gabriel_hscc,delimpaltadakis_cdc_2020,delimpaltadakis2020abstracting} only non-stochastic systems have been considered. Specifically, \cite{arman_formal_etc,gleizer2020scalable,gabriel_hscc} consider linear ETC and PETC systems, whereas \cite{delimpaltadakis_cdc_2020} and \cite{delimpaltadakis2020abstracting} address nonlinear systems with bounded disturbances. In this work, we consider \emph{stochastic} systems, for the first time. 
	Among others, for verification purposes, the probabilistic framework of stochastic systems is naturally less strict than the deterministic one, as it takes into account the disturbances' probability distribution, instead of being bound by worst case scenarios.
	
	In particular, we consider stochastic narrow-sense linear PETC systems. We define their sampling behaviour as the set $\Y$ of all possible sequences of state-measurements and intersampling times along with its associated probability measure. Studying ETC's sampling behaviour is formalized by computing expectations of functions defined over these sequences $g:\Y\to \real$. Here, we focus on functions $g$ described as cumulative, average or multiplicative rewards, i.e. $g_\star$ with $\star\in\{\cum,\avg,\mul\}$. This class of functions is rather standard in the context of quantitative analysis of stochastic systems, and it extends to including specifications of PCTL (Probabilistic Computation Tree Logic, see \cite{morteza2019imdp}). Besides, it is able to describe various metrics on ETC's sampling performance, as demonstrated through examples. In fact, as shown in one example, including state-measurements in the sampling behaviour's definition allows for incorporating control-performance metrics as well. The problem statement of this work is to obtain bounds on expectations of functions $g_\star$. 
	
	To address the problem, we construct IMCs (interval Markov chains; Markov chains with interval transition probabilities) that capture PETC's sampling behaviour. Then, we equip the IMCs with appropriate state-dependent rewards and prove that the $\{\cum,\avg,\mul\}$ reward over the paths of the IMC indeed bounds the expectation of $g_\star$ (Theorem \ref{main_theorem}). The IMC rewards can easily be computed via well-known value-iteration algorithms (see, e.g., \cite{givan_bmdps}). 
	
	The main challenge in constructing the IMC is computing the IMC's probability intervals. For that, we study the joint probabilities of transitioning from one region of the state-space to another one with the intersampling time taking a specific value. Computation of these probabilities is more complicated than the traditional transition probabilities that appear in the literature of IMC-abstractions (e.g., \cite{lahijanian2015dt_imcs,coogan2020,luca_tac_2021,jackson2021strategy}), due to the presence of intersampling time as an event. To cope with that, we employ a series of convex relaxations and the fact that the system's state is a Gaussian process. That way, we reformulate computing these probabilities as optimization problems of log-concave objective functions and hyperrectangle constraint sets, which are easy to solve. Finally, our results are demonstrated through a numerical example, where we compute bounds on the expected average intersampling time and on the probability of triggering with the maximum possible intersampling time in a finite horizon.
	
	In summary this work's main contributions are:
	\begin{itemize}
		\item It is the first one to abstract the sampling behaviour of stochastic ETC. 
		\item It computes bounds on performance metrics over stochastic PETC's sampling behaviour, allowing for its formal assessment and prediction of its patterns.
	\end{itemize}
	A preliminary version of the present work was presented in \cite{delimpaltadakis2021abstracting}. In \cite{delimpaltadakis2021abstracting}, only cumulative rewards are addressed and the derived upper and lower bounds on transition probabilities are different (here, they are tighter). Furthermore, the proof of Lemma \ref{lemma:log-concave}, which shows log-concavity of our optimization problems' objective functions, appears here for the first time. Finally, the proof of Theorem \ref{main_theorem} here is more elaborate than the proof of \cite[Theorem IV.1]{delimpaltadakis2021abstracting}; it argues about any horizon employing time-varying adversaries, whereas \cite[Theorem IV.1]{delimpaltadakis2021abstracting} argues about infinite-horizons, where time-invariant adversaries suffice.

	\section{Preliminaries}
	\subsection{Notation}
	$\real$ stands for the set of real numbers, $\mathbb{N}$ for the natural numbers including 0, and $\mathbb{N}_+$ without 0. Given $X\subseteq\real$, $X_{[a,b]}=X\cap[a,b]$. $I_n$ is the $n$-dimensional identity matrix. Given a set $S$ in some space $X$, we denote: its indicator function by $\ind_S(\cdot)$, its Borel $\sigma$-algebra by $\mathcal{B}(S)$, its complement by $\overline{S} = X\setminus S$, and the $k$-times Cartesian product $S=S\times\dots\times S$ by $S^k$. 
	Given $x\in\real^n$, denote by $\{x\}^k$: both the $k$-times Cartesian product $\{x\}\times\dots\times\{x\}$ and the $kn$-dimensional vector $\begin{bmatrix}
		x^\top &\dots &x^\top
	\end{bmatrix}^\top$. Given sets $Q_1,Q_2$ and $Q=Q_1\times Q_2$, for any $q=(q_1,q_2)\in Q$ denote $\proj_{Q_1}(q)=q_1$ and $\proj_{Q_2}(q)=q_2$. Given two sets $Q_1,Q_2$ in some space, denote $Q_1+Q_2 = \{q_1+q_2:q_1\in Q_1,q_2\in Q_2\}$ (Minkowski sum) and $Q_1-Q_2 = \{q_1-q_2:q_1\in Q_1,q_2\in Q_2\}$ (Minkowski difference). Finally, consider a set $S(x)\subseteq\real^n$ that varies with a parameter $x\in \real^m$ (equivalent to a set-valued function $S:\real^m\to2^{\real^n}$). 
	We say that $S(x)$ is \emph{linear on} $x$, if $S(x) = S'+\{Ax\}$, where $A\in\real^{n\times m}$ and $S'\subseteq\real^n$.
	
	Given a random variable $x$ and an associated probability measure $\pr$, we denote its expectation w.r.t. $\pr$ by $\expec_{\pr}[x]$ (when $\pr$ is clear from the context, it might be omitted). We use the term `path' or `sequence' interchangeably. Given a finite path $\omega = q_0,q_1\dots,q_N$, denote $\omega(i)=q_i$ and $\omega(end)=\omega(N)=q_N$. Given a function $g(\omega)$ of paths $\omega$, we denote $\expec^{q_0}[g(\omega)]\equiv\expec[g(\omega)|\omega(0)=q_0]$. Finally, $\normal(\mu,\Sigma)$ denotes the Gaussian distribution with mean $\mu$ and covariance matrix $\Sigma$.
	
	\subsection{Rewards over Paths} \label{sec:rewards}
	Consider a set $Q$ and a set of paths $\Y$ of length $N+1$, such that: $\omega(i)\in Q$, for all $\omega\in\Y$ and $0\leq i \leq N$. Assume a probability measure $\pr$ over $\borel(\Y)$ (for how to define $\borel(\Y)$ in our context, see Section \ref{sec:sampling_behaviour}). Define a \emph{reward function} $R:Q\to[0,R_\max]$. 
	We define the following expectations:
	\begin{itemize}
		\item \textit{Cumulative (discounted) reward}: $\expec_{\pr}[g_{\cum,N}(\omega)]\equiv\expec_{\pr}[\sum_{i=0}^{N}\gamma^iR(\omega(i))]$, where $\gamma\in[0,1]$.
		\item \textit{Average reward}: $\expec_{\pr}[g_{\avg,N}(\omega)]\equiv\expec_{\pr}[\tfrac{1}{N+1}\sum_{i=0}^{N}R(\omega(i))]$.
		\item \textit{Multiplicative reward}: $\expec_{\pr}[g_{\mul,N}(\omega)]\equiv\expec_{\pr}[\prod_{i=0}^{N}R(\omega(i))]$.
	\end{itemize}
	These expectations can describe a wide range of quantitative/qualitative properties of paths in $\Y$, and they have been employed for verification in numerous settings, such as (interval) Markov chains (e.g. \cite{lahijanian2015dt_imcs,coogan2020,luca_tac_2021,jackson2021strategy}), stochastic hybrid systems (e.g., \cite{abate2008probabilistic}), etc. Later, we showcase their descriptive power within our framework (see Section \ref{sec:sampling_behaviour}). 
	
	\subsection{Interval Markov Chains (IMCs)}\label{sec:prelim_imcs}
	Interval Markov Chains are Markov models with interval transition probabilities, 
	and they are defined as:
	\begin{definition}[Interval Markov Chain (IMC)]
		An IMC is a tuple $\sys_\imc=\{Q, \check{P} ,\hat{P}\}$, where: $Q$ is a finite set of states, and $\check{P},\hat{P}:Q\times Q\to[0,1]$ are functions, with $\check{P}(q,q')$ and $\hat{P}(q,q')$ representing lower and upper bounds on the probability of transitioning from state $q$ to $q'$, respectively.
	\end{definition}
	For all $q\in Q$, we have that $\check{P}(q,q')\leq \hat{P}(q,q')$ and $\sum\limits_{q'\in Q}\check{P}(q,q')\leq 1 \leq \sum\limits_{q'\in Q}\hat{P}(q,q')$. A path of an IMC is a sequence of states $\omega = q_0, q_1, q_2,\dots$, with $q_i\in Q$. Denote the set of the IMC's finite paths by $Paths^{fin}(\sys_{\imc})$. Given a state $q\in Q$, a transition probability distribution $p_{q}:Q\to[0,1]$ is called \textit{feasible} if $\check{P}(q,q')\leq p_{q}(q')\leq\hat{P}(q,q')$ for all $q'\in Q$. Given $q\in Q$, its set of feasible distributions is denoted by $\Gamma_q$. We denote by $\Gamma_Q=\{p_q:p_q\in\Gamma_q, q\in Q\}$ the set of all feasible distributions for all states.
	\begin{definition}[Adversary]
		Given an IMC $\sys_{\imc}$, an \emph{adversary} is a function $\adv:Paths^{fin}(\sys_\imc)\to\Gamma_Q$, such that $\adv(\omega)\in\Gamma_{\omega(end)}$, i.e. given a finite path it returns a feasible distribution w.r.t. the path's last element.
	\end{definition}
	The set of all adversaries is denoted by $\Pi$. Given a $\adv\in\Pi$ and $\omega(0) = q_0$, an IMC path evolves as follows: at any time-step $i>0$ $\adv$ chooses a distribution $p\in\Gamma_{\omega(i-1)}$ from which $\omega(i)$ is sampled.
	
	IMCs may be equipped with a reward function $R:Q\to[0,R_{\max}]$. Given a $\adv\in\Pi$ and an initial condition $q_0\in Q$, all expectations listed in Section \ref{sec:rewards} are well-defined and single-valued: e.g., $\expec^{q_0}_\adv[g_{\cum,N}(\omega)]$ (see \cite{givan_bmdps}). However, due to the existence of infinite adversaries, the IMC produces whole ranges of such expectations. The bounds of these ranges, e.g. ($\sup_{\adv\in\Pi}$ and) $\inf_{\adv\in\Pi}\expec^{q_0}_\adv[g_{\cum,N}(\omega)]$, can be computed via well-known \emph{value iteration} algorithms (e.g., see \cite{givan_bmdps,puterman}). 

	\section{The Sampling Behaviour of Stochastic PETC: Framework and Problem Statement}
	\subsection{Linear Stochastic PETC Systems}\label{sec:petc}
	Consider a state-feedback stochastic linear control system:
	\begin{equation*}
		d\zeta(t) = A\zeta(t)dt + BK\zeta(t)dt +B_wdW(t),
	\end{equation*}
	where: 
	$A,B,K,B_w$ are matrices of appropriate dimensions, $\zeta(t)\in\real^{n_\zeta}$ is the state, and $W(t)$ is an $n_w$-dimensional Wiener process on a complete probability space $(\Omega,\mathcal{F}, \{\algeb_t\}_{t\geq0}, \pr)$. $\Omega$ denotes the sample space, $\algeb$ the $\sigma$-algebra generated by $W$, $\{\algeb_t\}_{t\geq0}$ the natural filtration and $\pr$ the probability measure. We denote the solution of the above stochastic differential equation with initial condition $\zeta_0$ by $\zeta(t;\zeta_0)$.
	
	In PETC, the control input is held constant between consecutive \emph{sampling times} (or \emph{event times}) $t_i,t_{i+1}$ and is only updated on such times:
	\begin{equation}\label{snh}
		d\zeta(t) = A\zeta(t)dt + BK\zeta(t_i)dt +B_wdW(t), \quad t\in[t_i,t_{i+1}),
	\end{equation}
	Sampling times are determined by the \emph{triggering condition}:
	\begin{equation}\label{trig_cond}
		\begin{aligned}
			t_{i+1} = t_i + &\min\bigg\{\overline{k} h,\\&\min\Big\{kh:k\in\mathbb{N}, \phi\Big(\zeta(kh;\zeta(t_i)),\zeta(t_i)\Big)> 0\Big\}\bigg\}
		\end{aligned}
	\end{equation}
	where $t_0=0$, $h>0$ is a \textit{checking period}, $\overline{k}\in\mathbb{N}_{+}$, $\phi$ is called \textit{triggering function} and $t_{i+1}-t_i$ is called \textit{intersampling time}. PETC works as follows during an intersampling interval $[t_i,t_{i+1})$: at time $t_i$ the triggering function $\phi(\zeta(t_i),\zeta(t_i))$ is negative; the sensors check periodically, with period $h$, if the triggering function is positive; if it is found positive, or if $\overline{k} h$ time has elapsed since $t_i$, a new event $t_{i+1}$ is triggered, the latest state-measurement is sent to the controller which updates the control action, and the whole process is repeated again. The forced upper-bound $\overline{k} h$ on intersampling times prevents the system from operating open-loop indefinitely. We call the combination \eqref{snh}-\eqref{trig_cond} \textit{(stochastic) PETC system}.
	
	Intersampling time is a random variable that depends on the previously measured state and we denote it as follows:
	\begin{equation*}
		\tau(x) = \min\bigg\{\overline{k} h,\min\Big\{kh:k\in\mathbb{N}, \phi\Big(\zeta(kh;x),x\Big)> 0\Big\}\bigg\}
	\end{equation*} 
	where $x\in\real^n$ is the previously measured state. Note that, because the system is time-homogeneous, reasoning w.r.t. the interval $[t_i,t_{i+1})$ is equivalent to reasoning w.r.t. $[0, t_{i+1}-t_i)$.
	\begin{assumption}\label{assum1}
		We assume the following:
		\begin{enumerate}
			\item The matrix pair $(A,B_w)$ is controllable.\label{assum1_controllability}
			\item The checking period $h=1$.\label{assum1_period}
			\item 
			$\phi(\zeta(t;x),x) = |\zeta(t;x)-x|_\infty-\epsilon$, where $\epsilon>0$ is a predefined constant.\label{assum1_trig_fun}
		\end{enumerate}
	\end{assumption}
	Item \ref{assum1_controllability} guarantees that $\zeta(t)$ is a non-degenerate Gaussian random variable (see \cite{luca_tac_2021}). Item \ref{assum1_period} is for ease of presentation and without loss of generality. Regarding item \ref{assum1_trig_fun}, $\phi$ is the well-studied Lebesgue-sampling triggering function \cite{astrom2002comparison} with an $\infty$-norm instead of a $2$-norm. We restrict to this case for clarity, but our results are extendable to more general functions.
	\begin{remark}\label{remark_lebesgue_proof}
		Modifying the proof of {\cite[Theorem 1]{on_etc_stochastic_tac2020}}, one can show that Lebesgue-sampling guarantees mean-square practical stability for PETC system \eqref{snh}-\eqref{trig_cond}.
	\end{remark}
	
	\subsection{Sampling Behaviour and Associated Metrics}\label{sec:sampling_behaviour}
	A stochastic PETC system may exhibit different sequences of state-measurements and intersampling times $(\zeta_0,t_0),(\zeta(t_1),t_1-t_0),(\zeta(t_2),t_2-t_1),\dots$, where $t_i$ are sampling times. 
	We call \emph{sampling behaviour}, the set of all possible such sequences:
	\begin{definition}[Sampling Behaviour] \label{def:sampling behaviour}
		We call \emph{$N$-sampling behaviour} of stochastic PETC system \eqref{snh}-\eqref{trig_cond} the set:
		\begin{equation}
			\begin{aligned}
				\Y_N=\{(x_0,s_0),(x_1,s_1),(x_2,s_2)&,\dots,(x_N,s_N)|\\& x_i\in\real^{n_\zeta}, s_i\in\mathbb{N}_{[0, \overline{k}]}\}
			\end{aligned}
		\end{equation}
		where $N\in\mathbb{N}$. When $N$ is clear from the context, it is omitted.
	\end{definition}	  
	We denote $Q:=\real^{n_\zeta}\times\mathbb{N}_{[0,\overline{k}]}$. Given an initial condition $y_0 = (x_0,s_0)\in Q$, the set $\Y_N$ is associated to a probability measure $\pr^{y_0}_{\Y_N}$ (conditioned on $y_0$) which is inductively defined over $\borel(\Y_N)$ as follows\footnote{Consider $Q^{N+1}$ endowed with its product topology. Then $\borel(\Y_N)$ is the $\sigma$-algebra generated by cylinder sets of $Q^{N+1}$.}:
\begin{align}
	&\pr^{y_0}_{\Y_N}(\omega(0)\in (X_0,s_0)) =\ind_{(X_0,s_0)}(y_0) \label{eq:prob measure1}\\
	&\pr^{y_0}_{\Y_N}(\omega(i+1)\in (X_{i+1},s_{i+1})\text{ }|\text{ }\omega(i)=(x_i,s_i)) =\nonumber\\ &\pr(\zeta(s_{i+1};x_i)\in X_{i+1}, \tau(x_i)=s_{i+1})\label{eq:prob measure2}
\end{align}
where $\omega\in\Y_N$,
$s_0,s_i,s_{i+1}\in\mathbb{N}_{[0,\overline{k}]}$, $x_i\in\real^{n_\zeta}$, $X_0,X_{i+1}\subseteq\real^{n_\zeta}$ and we use $(X,s)$ to denote the set $\{(x,s):x\in X\}$. This measure is well-defined, even when the horizon $N=+\infty$, according to the Ionescu-Tulcea theorem \cite{ionescu}. 

\begin{remark}\label{rem:init_cond}
	As noted in Section \ref{sec:petc}, typically it is assumed that the first sampling time $t_0=0$, which implies that the first intersampling time $s_0=t_0-0=0$ and the initial condition is $y_0=(x_0,0)$. 
\end{remark}
\begin{remark}\label{rem:zero_tau}
	Under item \ref{assum1_trig_fun} of Assumption \ref{assum1}, and in every Zeno-free ETC scheme, $\pr(\tau(x)=0) = 0$ for any $x\in\real^{n_\zeta}$, because the triggering function is strictly negative for $k=0$. Thus, for any $i\geq 1$ and $\omega\in \Y_N$: $\pr_{\Y_N}(\proj_{\mathbb{N}_{[0,\overline{k}]}}(\omega(i))=0)=0$. Note that this is not in contrast with Remark \ref{rem:init_cond} that only reasons about initial conditions $(x_0,s_0)$ and not $(x_i,s_i)$ with $i\geq 1$.
\end{remark}

Studying PETC's sampling behaviour may be formalized by defining functions $g:\Y_N\to\real$ and computing their expectations $\expec_{\pr^{y_0}_{\Y_N}}[g(\omega)]$. Here, we focus on functions that can be described as cumulative $g_{\cum,N}$, average $g_{\avg,N}$ or multiplicative $g_{\mul,N}$ rewards (see Section \ref{sec:rewards}). By appropriately choosing the reward $R$, these classes of functions can describe many interesting properties of PETC's sampling behaviour:
\begin{itemize}
	\item \textit{Example 1:} Consider $R(x,s)=s$. Then $\expec_{\pr^{y_0}_{\Y_N}}[g_{\avg,N}(\omega)]$ is the expected average intersampling time: the larger it is, the less frequently the system is expected to sample, saving more bandwidth and energy.
	\item \textit{Example 2:} Consider $R(x,s)=\min(\alpha\tfrac{1}{|x|+\varepsilon}+\beta s,R_\max)$, with $\alpha,\beta,\varepsilon>0$, penalizing paths that overshoot far from the origin or exhibit a high sampling frequency. A bigger $\expec_{\pr^{y_0}_{\Y_N}}[g_{\cum,N}(\omega)]$ implies better performance in terms of stabilization speed and sampling frequency. Observe how incorporating state-measurements $x$ in our definition of sampling behaviour, allows to include control-performance related metrics, apart from sampling-performance metrics.
	\item \textit{Example 3:} Consider the reward:
	\begin{equation*}
		R(x,s) = \left\{\begin{aligned}
			&0,\quad \text{if }s=\overline{k}\\
			&1, \quad \text{otherwise}
		\end{aligned}\right.
	\end{equation*}
	Then, we have that:
	\begin{align*}
		\expec_{\pr^{y_0}_{\Y_N}}[g_{\mul,N}(\omega)] &= \pr^{y_0}_{\Y_N}\Big(\proj_{\mathbb{N}_{[0,\overline{k}]}}(\omega(i))\neq\overline{k}, \text{ }\forall i \Big)
	\end{align*}
	$\expec_{\pr^{y_0}_{\Y_N}}[g_{\mul,N}(\omega)]$ is the probability that there is no intersampling time $s=\overline{k}$ in the next $N$ events. 
	The smaller it is, the more probable it is that the system samples, at least once in the first $N$ triggers, with intersampling time $s=\overline{k}$, implying that a bigger maximum intersampling time could be used, allowing the system to sample even less frequently and saving more bandwidth. 
\end{itemize}

Observe that, if the initial condition $(x_0,s_0)$ is only known to obey some distribution $p_0:Q\to[0,1]$, the expected reward can be described as:
\begin{equation*}
	\expec_{\pr^{p_0}_{\Y_N}}[g_{\star,N}(\omega)] = \sum_{s_0\in\mathbb{N}_{[0,\overline{k}]}}\hspace{-2mm}\int_{\real^{n_\zeta}}\expec_{\pr^{(x_0,s_0)}_{\Y_N}}[g_{\star,N}(\omega)]p_0(x_0,s_0)dx_0
\end{equation*}
Thus, reasoning about individual initial conditions $y_0$ is sufficient and immediately extends to the general case of random initial conditions.

Overall, defining PETC's sampling behaviour $\Y_N$, associating it to its induced probability measure $\pr^{y_0}_{\Y_N}$ given in \eqref{eq:prob measure1}-\eqref{eq:prob measure2}, and studying expectations $\expec_{\pr^{y_0}_{\Y_N}}[g(\omega)]$ constitutes a formal framework for the study of PETC's sampling behaviour.  

\subsection{Problem Statement}\label{sec:prob_stat}
Unfortunately, exact computation of $\expec_{\pr^{y_0}_{\Y_N}}[g_{\star,N}(\omega)]$ is generally infeasible. Among others, how to obtain the measure $\pr^{y_0}_{\Y_N}$ over the \emph{uncountable} set of paths $\Y_N$ and then integrate over it? Hence, we aim at computing bounds over such expectations:
\begin{problem statement}
	Consider the PETC system \eqref{snh}-\eqref{trig_cond} and its sampling behaviour $\Y_N$, for some $N\in\mathbb{N}$. Let Assumption \ref{assum1} hold. Consider a reward function $R:Q\to[0,R_\max]$. For all initial conditions $y_0\in X\times\mathbb{N}_{[0,\overline{k}]}$, where $X\subset\real^{n_\zeta}$ is compact, 
	compute (non-trivial) lower and upper bounds on $\expec_{\pr^{y_0}_{\Y_N}}[g_{\star,N}(\omega)]$, where $\star\in\{\cum,\avg,\mul\}$. 
\end{problem statement}
\begin{figure*}[t!]
	\begin{subfigure}[t]{\linewidth}
		\centering
		\includegraphics[width = 0.95\linewidth]{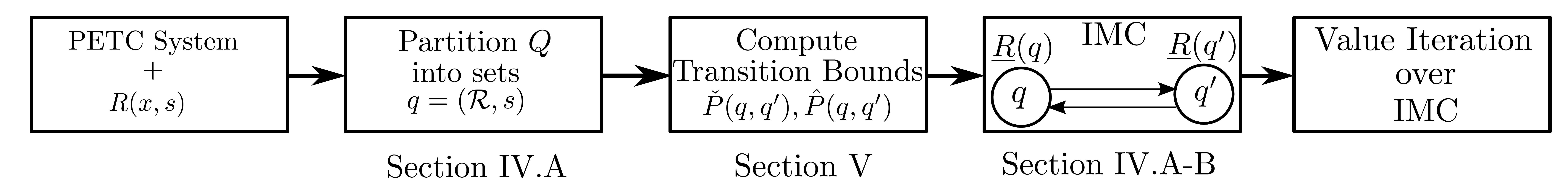}
	\end{subfigure}
	\caption{A flowchart showing the steps followed to compute bounds on the expected rewards $\expec_{\pr^{y_0}_{\Y_N}}[g_{\star,N}(\omega)]$.}
	\label{fig:flowchart}
\end{figure*}
In the rest of this article, we address the problem by constructing an IMC that \emph{abstracts} the sampling behaviour $\Y_N$ along with $\pr^{y_0}_{\Y_N}$, equipping it with suitable reward functions $\underline{R},\overline{R}$, and computing ($\sup_{\adv\in\Pi}$ and) $\inf_{\adv\in\Pi}\expec^{q_0}_\adv[g_{\star,N}(\tilde{\omega})]$, with $\star\in\{\cum,\avg,\mul\}$, to obtain the bounds we are looking for. Specifically, in the next section we show how to construct such an IMC, by partitioning the state space and providing conditions (eq. \eqref{eq:probability_bounds}-\eqref{eq:unsafe_state_probabilities}) that have to be satisfied by the IMC's transition probability intervals. We prove in Theorem \ref{main_theorem} that this IMC equipped with suitable rewards gives rise to bounds on $\expec_{\pr^{y_0}_{\Y_N}}[g_{\star,N}(\omega)]$. Later, in Section \ref{sec:transitions}, we show how to compute $\check{P}$ and $\hat{P}$ such that they satisfy \eqref{eq:probability_bounds}-\eqref{eq:unsafe_state_probabilities}, by solving optimization problems with log-concave objective functions. Finally, the desired bounds ($\sup_{\adv\in\Pi}$ and) $\inf_{\adv\in\Pi}\expec^{q_0}_\adv[g_{\star,N}(\tilde{\omega})]$ are obtained via well-known value iteration algorithms, as demonstrated through a numerical example in Section \ref{sec:examples}. A flowchart of the steps followed to compute the desired bounds is shown in Figure \ref{fig:flowchart}.

\begin{remark} 
	By assuming that $y_0 = (x_0,s_0)\in X\times\mathbb{N}_{[0,\overline{k}]}$, we essentially assume that the initial state of the system $x_0\in X$. Compactness of $X$ is vital, to partition it into a finite number of subsets $\reg_i$ and end up with a finite-state IMC. Nonetheless, this is not an unrealistic assumption, as in practice the initial conditions of the system are usually known to be bounded in some set. Furthermore, $s_0\in\mathbb{N}_{[0,\overline{k}]}$ for generality, but, as mentioned in Remark \ref{rem:init_cond}, typically in ETC $s_0=0$.
\end{remark}
\begin{remark} \label{rem:pctl}
	We constrain ourselves to $\{\cum,\avg,\mul\}$ rewards for clarity, but our approach extends to a more general framework. As commented in Section \ref{sec:underline_R}, our IMCs can be employed for computing bounds on \emph{bounded-until} probabilities:
	\begin{equation*}
		\pr^{y_0}_{\Y_N}(\exists i\in\mathbb{N}_{[0,N]} \text{ s.t. }\omega(i)\in G \text{ and }\forall k\leq i, \text{ }\omega(k)\in S)
	\end{equation*} 
	where $S,G\subseteq Q$. 
	Bounded-until constitutes the backbone of PCTL \cite{lahijanian2015dt_imcs}, as all PCTL formulas can be written with bounded-until operations. In short, our approach directly extends to PCTL. Moreover, by extending our proofs according to \cite{coogan2020} and \cite{jackson2021strategy}, we could incorporate probabilistic $\omega-$regular or LTL (Linear Temporal Logic) properties. 
\end{remark} 

\section{IMCs Abstracting PETC's Sampling Behaviour}

\subsection{Constructing the IMC} \label{sec:imc}
Typically, to abstract a stochastic behaviour $\Y$ and its probability measure $\pr_{\Y}$ through an IMC: i) the state space is partitioned into a finite number of regions, each of which corresponds to an IMC-state, ii) if the state space is unbounded, then one of these regions is unbounded, and its IMC-state is made absorbing\footnote{An IMC-state $q\in Q_{\imc}$ is absorbing $\iff$ $\check{P}(q,q)=1$.}, and iii) the bounds on transition probabilities $\check{P}(q,q'),\hat{P}(q,q')$ are derived such that $\check{P}(q,q')\leq\pr_{\Y}(\omega(i+1)\in q' | \omega(i) = x)\leq\hat{P}(q,q')$ for all $x\in q$, where $\omega \in \Y$.

In this work, we adopt the above methodology. Observe that the state space from which the sampling behaviour emerges is the set $Q = \real^{n_\zeta} \times \mathbb{N}_{[0,\overline{k}]}$. Since $\mathbb{N}_{[0,\overline{k}]}$ is by-construction partitioned into the singletons $\{0\},\{1\},\dots,\{\overline{k}\}$, it suffices to partition $\real^{n_\zeta}$. Consider 
$m$ non-overlapping 
compact regions $\reg_i$ such that $\bigcup_{i\in\mathbb{N}_{[1,m]}}\reg_i = X$. Then, $\real^{n_\zeta}$ is partitioned into:
\begin{equation*}
	Q_\reg\cup\{\overline{X}\}
\end{equation*}
where $Q_\reg = \{\reg_1,\dots,\reg_m\}$. According to the aforementioned methodology, the states of the IMC would be of the form $(q,s)\in\Big(Q_\reg\cup\{\overline{X}\}\Big)\times\mathbb{N}_{[0,\overline{k}]}$. Nonetheless, for compactness of the IMC, we group all states $(\overline{X},s)$ (for $s\in\mathbb{N}_{[0,\overline{k}]}$) that correspond to $\overline{X}$ into a single absorbing state $\uns$:
\begin{equation*}
	\uns = \overline{X}\times\mathbb{N}_{[0,\overline{k}]}
\end{equation*} 
From now on, we abusively use $(\reg_i,s)$ (resp. $\uns$) to denote both the corresponding IMC-state and the set $\reg_i\times\{s\}$ (resp. $\overline{X}\times\mathbb{N}_{[0,\overline{k}]}$). Finally, the set of the IMC-states is:
\begin{equation}\label{eq:qimc}
	Q_{\imc} = \Big(Q_\reg \times \mathbb{N}_{[0,\overline{k}]}\Big)\cup\{\uns\}
\end{equation}

Regarding the transition probability bounds $\check{P}(q,q')$ and $\hat{P}(q,q')$, since we need to bound $\pr(\omega(i+1)\in q' | \omega(i) = x)$ for all $x\in q$, by employing \eqref{eq:prob measure2}, we have that for all $(\reg,k),(\sreg,s)\in Q_\reg \times \mathbb{N}_{[0,\overline{k}]}$:
\begin{equation}\label{eq:probability_bounds}
	\begin{aligned}
		&\check{P}\Big((\reg,k),(\sreg,s)\Big)\leq \min\limits_{x\in\reg}\pr(\zeta(s;x)\in \sreg,\tau(x)=s)\\
		&\hat{P}\Big((\reg,k),(\sreg,s)\Big)\geq \max\limits_{x\in\reg}\pr(\zeta(s;x)\in \sreg,\tau(x)=s)\\
		&\check{P}\Big((\reg,k),\uns\Big)\leq\sum\limits_{s\in\mathbb{N}_{[0,\overline{k}]}}\hspace{-3mm}\min\limits_{x\in\reg}\pr(\zeta(s;x)\in \overline{X},\tau(x)=s)\\
		&\hat{P}\Big((\reg,k),\uns\Big)\geq\sum\limits_{s\in\mathbb{N}_{[0,\overline{k}]}}\hspace{-3mm}\max\limits_{x\in\reg}\pr(\zeta(s;x)\in \overline{X},\tau(x)=s)
	\end{aligned}
\end{equation}
and for all $q'\in Q_\imc$:
\begin{equation}\label{eq:unsafe_state_probabilities}
	\check{P}(\uns,q')=\hat{P}(\uns,q') = \left\{\begin{aligned}
		&1,\text{ if }q'=\uns\\ &0,\text{ otherwise}
	\end{aligned}\right.
\end{equation}
The computation of $\check{P}$ and $\hat{P}$ such that they satisfy \eqref{eq:probability_bounds}-\eqref{eq:unsafe_state_probabilities} is addressed in Section \ref{sec:transitions} and it involves bounding the solutions to the optimization problems of \eqref{eq:probability_bounds}.
The summation in the last two inequalities of \eqref{eq:probability_bounds} results from the fact that $\uns$ is a grouping of all states $(\overline{X},s)$ with $s\in\mathbb{N}_{[0,\overline{k}]}$, while \eqref{eq:unsafe_state_probabilities} indicates that $\uns$ is indeed absorbing. In view of Remark \ref{rem:zero_tau}, since we know that $\pr(\tau(x)=0) = 0$, then for any $q\in Q_\imc$ and $\sreg\in Q_\reg$, it suffices to write $\check{P}(q,(\sreg,0))=\hat{P}(q,(\sreg,0))=0$; that is, states $(\sreg,0)$ only have outgoing transitions and no incoming ones. Finally, we define the IMC that abstracts the sampling behaviour as follows:
\begin{equation}\label{eq:our_imc}
	\sys_\imc =(Q_\imc,\check{P},\hat{P}),
\end{equation}
where $Q_\imc$ is given by \eqref{eq:qimc} and $\check{P},\hat{P}$ are given by \eqref{eq:probability_bounds}-\eqref{eq:unsafe_state_probabilities}. 

To demonstrate how the constructed IMC abstracts the PETC system's sampling behaviour, let us relate paths $\omega\in\Y_N$ to paths $\tilde{\omega}\in Paths^{fin}(\sys_{\imc})$. First, consider a path $\omega$ such that $\omega(i)\not\in\uns$ for all $i\leq N$. Then, this path is related to a path $\tilde{\omega}\in Paths^{fin}(\sys_{\imc})$ of the same length, for which $\omega(i)\in\tilde{\omega}(i)$ for all $i\leq N$. Next, consider a path such that $\omega(i)\in\uns$ for some $i\leq N$ and $\omega(j)\not\in\uns$ for all $j<i$. Then, $\omega$ is related to $\tilde{\omega}\in Paths^{fin}(\sys_{\imc})$ of the same length, for which $\omega(j)\in\tilde{\omega}(j)$ for all $j\leq i$ and $\tilde{\omega}(k)=\uns$ for all $k\geq i$. This latter relation indicates that all paths in $\Y_N$ that enter $\overline{X}$ (even those that eventually return to $X$) are mapped to IMC-paths that enter $\uns$ at the same time and stay there.

\subsection{Bounds on Sampling-Behaviour Rewards via IMCs}\label{sec:underline_R}
The IMC described above, if equipped with suitable rewards $\underline{R},\overline{R}$, can be employed for the computation of lower and upper bounds on $\expec_{\pr^{y_0}_{\Y_N}}[g_{\star,N}(\omega)]$:
\begin{theorem}\label{main_theorem}
	Consider the IMC $\sys_\imc$ given by \eqref{eq:our_imc}.
	Define reward functions $\underline{R},\overline{R}:Q_\imc\to[0,R_\max]$ such that:
	\begin{equation}\label{eq:underline_R}
		\begin{aligned}
			&\underline{R}(q) = \left\{\begin{aligned}
				&\min\limits_{(x,s)\in q}R(x,s), \text{ if }q\neq\uns\\&\min\limits_{(x,s)\in \real^{n_\zeta}\times\mathbb{N}_{[1,\overline{k}]}}\hspace{-7mm}R(x,s), \text{ if }q=\uns
			\end{aligned}\right.\\
			&\overline{R}(q) = \left\{\begin{aligned}
				&\max\limits_{(x,s)\in q}R(x,s), \text{ if }q\neq\uns\\&\max\limits_{(x,s)\in \real^{n_\zeta}\times\mathbb{N}_{[1,\overline{k}]}}\hspace{-7mm}R(x,s), \text{ if }q=\uns
			\end{aligned}\right.
		\end{aligned}
	\end{equation}
	and the associated 
	rewards over paths $\tilde{\omega}\in Paths^{fin}(\sys_\imc)$ denoted by $\underline{g}_{\star,N},\overline{g}_{\star,N}$, where $\star\in\{\cum,\avg,\mul\}$.
	Then, for any initial condition $y_0=(x_0,s_0)\in X\times\mathbb{N}_{[0,\overline{k}]}$ and $N\in\mathbb{N}$:
	\begin{equation*}
		\inf_{\adv\in\Pi}\hspace{-.5mm}\expec^{q_0}_\adv[\underline{g}_{\star,N}(\tilde{\omega})]\leq\expec_{\pr^{y_0}_{\Y_N}}[g_{\star,N}(\omega)]\leq\sup_{\adv\in\Pi}\hspace{-.5mm}\expec^{q_0}_\adv[\overline{g}_{\star,N}(\tilde{\omega})]
	\end{equation*}
	where $q_0$ is such that $y_0 \in q_0$.
\end{theorem}
\begin{proof}[Proof Sketch]
	The above expectations are written as \emph{value functions} defined via \emph{value iteration} (see Lemma \ref{lemma:vi}), and mathematical induction over the iteration is employed. For the full proof, see Appendix \ref{app:main_theorem}.
\end{proof}
Hence, to compute bounds on expectations $\expec_{\pr^{y_0}_{\Y_N}}[g_{\star,N}(\omega)]$, we equip the IMC \eqref{eq:our_imc} with the reward functions $\underline{R},\overline{R}$ from \eqref{eq:underline_R} and compute the expectations $\inf_{\adv\in\Pi}\hspace{-.5mm}\expec^{q_0}_\adv[\underline{g}_{\star,N}(\tilde{\omega})]$ and $\sup_{\adv\in\Pi}\hspace{-.5mm}\expec^{q_0}_\adv[\overline{g}_{\star,N}(\tilde{\omega})]$. As mentioned in Section \ref{sec:prelim_imcs}, these expectations can be computed via value-iteration algorithms (e.g. see \cite{givan_bmdps}), with polynomial complexity in the number of IMC-states. In fact, the value iteration used for $\{\cum,\mul\}$ rewards is given here by equations \eqref{eq:viimc} and \eqref{eq:viimc_mul} respectively in the Appendix (the $\avg$ reward is the same as $\cum$ with $\gamma=1$, and in the last step we just divide by $N+1$). 
Moreover, since bounded-until probabilities on IMCs, and thus PCTL properties, may be computed through a similar value iteration \cite{lahijanian2015dt_imcs}, our proofs can be adapted to show that we can bound bounded-until probabilities defined over $\Y_N$ by using the constructed IMC. 

Finally, Theorem \ref{main_theorem} indicates that \emph{the same IMC} can be used to derive bounds for any chosen $\{\cum,\avg,\mul\}$ reward, for any horizon $N$, and any initial condition $y_0\in X \times \mathbb{N}_{[0,\overline{k}]}$. It is also worth noting that a proof like that of Theorem \ref{main_theorem} was missing from the literature on IMC-abstractions \cite{lahijanian2015dt_imcs,coogan2020,luca_tac_2021,jackson2021strategy}, where it was (correctly) taken for granted that the quantitative metric (e.g., a reward) evaluated over the IMC bounds the metric evaluated over the original stochastic behaviour, due to the way that the transition probabilities are constructed.
\begin{remark}\label{rem:partition_Y}
	For any $q\in Q_{\imc}$, the rewards $\underline{R}$ and $\overline{R}$ serve as conservative estimates of the real reward obtained if the system operates in $q$. In fact, specifically for $\uns$, $\underline{R}$ and $\overline{R}$ are \emph{global} lower and upper bounds, respectively, on the actual reward $R(x,s)$ (except for the case $s=0$, which happens with zero probability, except for initial conditions). 
	Due to this, for states $(\reg,s)\in Q_\imc$ with $\reg$ being ``near" $\overline{X}$ (i.e., near the boundary of $X$), which tend to obtain larger transition probabilities to $\uns$, the lower and upper bounds $\inf_{\adv\in\Pi}\hspace{-.5mm}\expec^{q_0}_\adv[\underline{g}_{\star,N}(\tilde{\omega})]$ and $\sup_{\adv\in\Pi}\hspace{-.5mm}\expec^{q_0}_\adv[\overline{g}_{\star,N}(\tilde{\omega})]$ are more conservative, compared to when $\reg$ is further inside $X$. This is showcased by Figure \ref{fig:surf_not_sample}. For that reason, in practice, to construct the IMC, it is better to partition a superset $Y\supseteq X$ into regions $\reg_i$, so that the regions that comprise $X$ are further inside $Y$, and the corresponding bounds are not that conservative.
\end{remark}
\begin{remark}
	Our results extend to infinite horizons (i.e. $N=+\infty$), when the rewards are well-defined, as it has already been proven in \cite[Theorem IV.1]{delimpaltadakis2021abstracting}; in fact, the proof for $N=+\infty$ is simpler, as it suffices to consider time-invariant adversaries.
\end{remark}
The only thing that remains is to describe how to compute the transition probability bounds given by \eqref{eq:probability_bounds}. This is carried out in the coming section.
\section{Computing the Transition Probability Bounds} \label{sec:transitions}
Here, we compute lower bounds on the minima and upper bounds on the maxima in \eqref{eq:probability_bounds}, thus completing the IMC's construction. Through a series of convex relaxations, and employing Proposition \ref{prop:normal} and Lemma \ref{lemma:log-concave}, the min/max expressions in \eqref{eq:probability_bounds} are formulated as optimization problems of log-concave functions (in fact, Gaussian integrals) over hyperrectangles, which are straightforward to solve. To facilitate this analysis, we introduce the following assumption:
\begin{assumption}
	The set $X$ and all sets $\reg_i\in Q_\reg$ are hyperrectangles.
\end{assumption}
This assumption is without loss of generality, as in the case where $X$ is not a hyperrectangle, our approach could be applied by under/overapproximating $X$ by a hyperrectangle $Y\supseteq X$ and partitioning $Y$ into a finite set of hyperrectangles $\reg_i$.

For the rest of the document, for any $s \in \mathbb{N}_{[1,N]}$, we denote $\zeta_{s,x}=\zeta(s;x)$ and $\tilde{\zeta}_{s,x} = [
	\zeta^\top_{1,x} \text{ }\zeta^\top_{2,x} \text{ }\dots \allowbreak\zeta^\top_{s,x}
]^\top$. The following statements are instrumental in our derivations:
\begin{proposition}\label{prop:normal}
	For any $s \in \mathbb{N}_{[1,N]}$, we have that $\tilde{\zeta}_{s,x}\sim\normal(\mu_{\tilde{\zeta}_{s,x}},\Sigma_{\tilde{\zeta}_{s,x}})$ with:
	\begin{align*}
		&\mu_{\tilde{\zeta}_{s,x}}=\begin{bmatrix}
			\expec(\zeta^\top_{1,x})&\expec(\zeta^\top_{2,x})&\dots&\expec(\zeta^\top_{s,x})
		\end{bmatrix}^\top\\
		&\Sigma_{\tilde{\zeta}_{s,x}}=\begin{bmatrix}
			\cov(1,1) &\cov(1,2)&\dots&\cov(1,s)\\
			\vdots &\vdots &\dots &\vdots\\
			\cov(s,1) &\cov(s,2) &\dots &\cov(s,s)
		\end{bmatrix}
	\end{align*}
	where $\expec(\zeta(t;x)) = [e^{At}(I+A^{-1}BK)-A^{-1}BK]x$ and:
	\begin{align*}
		&\cov(t_1,t_2) = \int_{0}^{\min(t_1,t_2)}e^{A(t_1-s)}B_wB_w^\top e^{A^\top(t_2-s)}ds
	\end{align*}
	Thus, given some set $S \subseteq \real^{sn_\zeta}$, the following holds:
	\begin{equation}
		\pr(\tzeta(s;x)\in S) = \int_S \normal(dz|\mu_{\tilde{\zeta}_{s,x}},\Sigma_{\tilde{\zeta}_{s,x}})
	\end{equation}
\end{proposition}
\begin{proof}
	Application of the expectation and covariance operators to the solution of linear SDE \eqref{snh} (see {\cite[pp. 96]{mao_book}}).
\end{proof}
\begin{lemma}\label{lemma:log-concave}
	Consider a function $h:\real^{n}\to[0,1]$ with $n\in\mathbb{N}_+$ defined by:
	\begin{equation*}
		h(x) = \int_{S(x)}\normal(dz|f(x),\Sigma)
	\end{equation*}
	where $\Sigma$ is a covariance matrix, $S(x)\subseteq \real^m$ with $m\in\mathbb{N}_+$ is linear on $x$ and convex for all $x\in\real^n$, and $f:\real^n\to\real^m$ is an affine function. The function $h(x)$ is log-concave on $x$.
\end{lemma}
\begin{proof}
	See Appendix \ref{sec:proofs_transitions_theorems}.
\end{proof}
In what follows we transform the probabilities involved in \eqref{eq:probability_bounds} to set-membership ones $\pr(\tzeta_{s,x}\in S(x))$, where $S(x)$ is a polytope, but neither necessarily convex nor linear on $x$. Afterwards, we break them down to simpler ones and employ some convex relaxations, such that the set of integration of the resulting Gaussian integrals is convex and linear on $x$ and Lemma \ref{lemma:log-concave} is enabled. Finally, we end up with optimization problems of log-concave functions over the hyperrectangle $\reg$, and solve them to obtain lower and upper bounds on the expressions in \eqref{eq:probability_bounds}. 

\subsection{Transition Probabilities as Set-Membership Probabilities}
For now, let us focus on transitions from any state $(\reg,k)\in Q_\imc\setminus\uns$ to any state $(\sreg,s)\in Q_\reg\times\mathbb{N}_{[1,\overline{k}]}$ :
\begin{equation*}
	(\max_{x\in\reg}\text{ or})\quad \min_{x\in\reg}\pr(\zeta(s;x)\in \sreg,\tau(x)=s)
\end{equation*} 
Later, in Section \ref{sec:transitions_to_uns}, we show how transitions to $\uns$ can be treated similarly to the case above. Moreover, remember that for $s=0$ the above probability is trivially 0 (see Remark \ref{rem:zero_tau}).  

Define the following hyperrectangle:
\begin{equation*}
	\Phi(x):=\{y\in\real^n: \phi(y,x)\leq 0\}=\{y\in\real^n:|y-x|_\infty\leq\epsilon\},
\end{equation*}
Note that $\Phi(x)$ is convex and linear on $x$: $\Phi(x) = \Phi(0) + \{x\}$. Moreover, it is such that $\zeta(t;x)\in\Phi(x)\iff\phi(\zeta(t;x),x)\leq 0$. Thus, the following equivalences hold:
\begin{equation*}
	\begin{aligned}
		&\text{if }s\in\mathbb{N}_{[1,\overline{k}-1]}: \tau(x) = s \iff \tilde{\zeta}_{s,x}\in\Phi^{s-1}(x)\times\overline{\Phi}(x)\\
		&\text{if }s=\overline{k}: \tau(x) = s=\overline{k} \iff \tilde{\zeta}_{\overline{k}-1,x}\in\Phi^{\overline{k}-1}(x)
	\end{aligned}
\end{equation*}
where, for brevity, in the case where $s=1$ we have abusively denoted $\Phi^0(x)\times\overline{\Phi}(x)=\overline{\Phi}(x)$. In words, when $s\neq \overline{k}$, the intersampling time is $s$ if and only if the state belongs to $\Phi(x)$ at all checking times $1,2,\dots,s-1$ and at time $s$ it lies outside $\Phi(x)$. When $s=\overline{k}$, it suffices that the state belongs to $\Phi(x)$ at all checking times $1,2,\dots,\overline{k}-1$. Thus, for $s\in\mathbb{N}_{[1,\overline{k}-1]}$:
\begin{equation}\label{eq:set-membership1}
	\begin{aligned}	
		\pr(\zeta(s;x)\in \sreg,\tau(x)=s) &= \pr\Big(\tilde{\zeta}_{s,x}\in\Phi^{s-1}(x)\times(\overline{\Phi}(x)\cap \sreg)\Big)\\
		&=\int_{\Phi^{s-1}(x)\times(\overline{\Phi}(x)\cap \sreg)}\hspace{-5mm} \normal(dz|\mu_{\tilde{\zeta}_{s,x}},\Sigma_{\tilde{\zeta}_{s,x}})
	\end{aligned}
\end{equation}
and for $s=\overline{k}$:
\begin{equation}\label{eq:set-membership2}
	\begin{aligned}
		\pr(\zeta(\overline{k};x)\in \sreg,\tau(x)=\overline{k}) &= \pr\Big(\tilde{\zeta}_{\overline{k},x}\in\Phi^{\overline{k}-1}(x)\times \sreg\Big)\\
		&=\int_{\Phi^{\overline{k}-1}(x)\times \sreg}\hspace{-5mm} \normal(dz|\mu_{\tilde{\zeta}_{\overline{k},x}},\Sigma_{\tilde{\zeta}_{\overline{k},x}})
	\end{aligned}
\end{equation}
In the following we combine \eqref{eq:set-membership1}-\eqref{eq:set-membership2} with some convex relaxations, to enable Lemma \ref{lemma:log-concave} and obtain bounds on $(\max_{x\in\reg}$ and) $\min_{x\in\reg}\pr(\zeta(s;x)\in \sreg,\tau(x)=s)$ through solving optimization problems with log-concave functions. In particular, observe that $\mu_{\tilde{\zeta}_{s,x}}$ is already an affine function of $x$ (see Proposition \ref{prop:normal}), thus satisfying one of the two conditions of Lemma \ref{lemma:log-concave}. Hence, our efforts focus on transforming the integration sets in \eqref{eq:set-membership1}-\eqref{eq:set-membership2} such that they become linear on $x$ and convex. 

\subsection{Lower Bounds on Transition Probabilities}
Let us start by determining lower bounds on:
\begin{equation} \label{eq:main_min_prob}
	\min_{x\in\reg}\pr(\zeta(s;x)\in \sreg,\tau(x)=s)
\end{equation}

The special case when $s=\overline{k}$, which is given by \eqref{eq:set-membership2}, is simple. Observe that the set $\Phi^{\overline{k}-1}(x)\times \sreg$ is convex (since $\Phi(x)$ and $\sreg\in Q_\reg$ are hyperrectangles) and linear on $x$, as it can be written as:
\begin{equation*}
	\begin{aligned}
		\Phi^{\overline{k}-1}(x)\times \sreg = \Phi^{\overline{k}-1}(0)&\times \sreg+ 	\begin{bmatrix}
			I_{n_\zeta} &I_{n_\zeta} &\dots &I_{n_\zeta} &0_{n_\zeta}
		\end{bmatrix}^\top \{ x\}
	\end{aligned}
\end{equation*}
Thus, when $s=\overline{k}$, the objective function $\pr(\zeta(s;x)\in \sreg,\tau(x)=s)$ of minimization problem \eqref{eq:main_min_prob} is log-concave (due to \eqref{eq:set-membership2} and Lemma \ref{lemma:log-concave}). The constraint set $\reg$ is a hyperrectangle. Thus, the minimization problem attains its solution at one of the vertices of $\reg$ \cite[pp. 343, Theorem 32.2]{rockafellar2015convex}; we simply have to evaluate the objective function for each of the vertices, to find the minimum.
\begin{figure*}[t!]
	\begin{subfigure}[t]{0.32\linewidth}
		\centering
		\includegraphics[height = 1.2in]{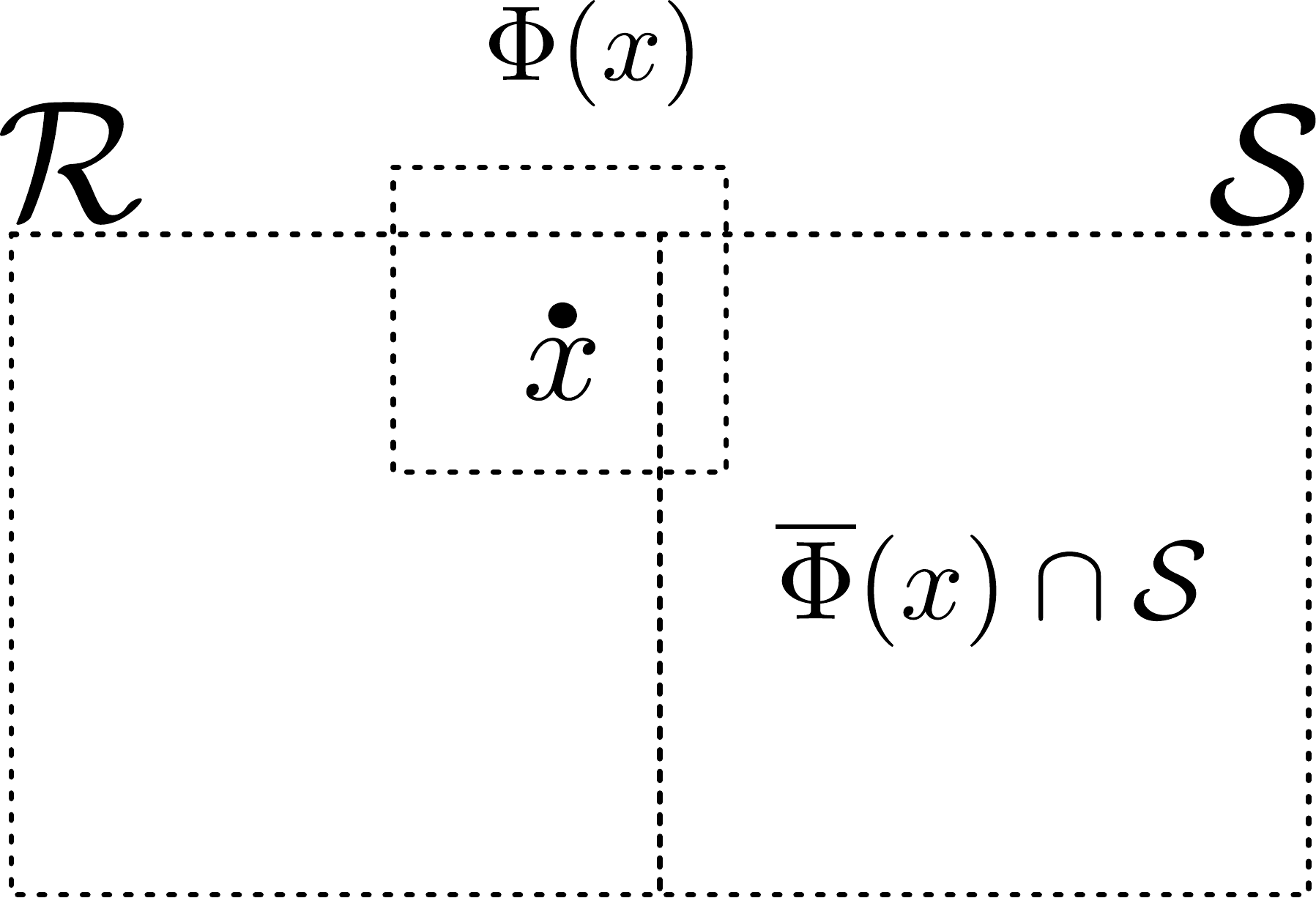}
		\caption{The set $\reg$ is the dashed square on the left, $\sreg$ is the one on the right, and $\Phi(x)$ is the one centered at $x$. For the given $x$ on the figure, $\overline{\Phi}(x)\cap \sreg$ is non-convex. For different $x\in\reg$ the set $\overline{\Phi}(x)\cap \sreg$ has a different shape; thus, $\overline{\Phi}(x)\cap \sreg$ is not linear on $x$.}
		\label{fig:nonconvex_intersection}
	\end{subfigure}~~~
	\begin{subfigure}[t]{0.32\linewidth}
		\centering
		\includegraphics[height = 1.2in]{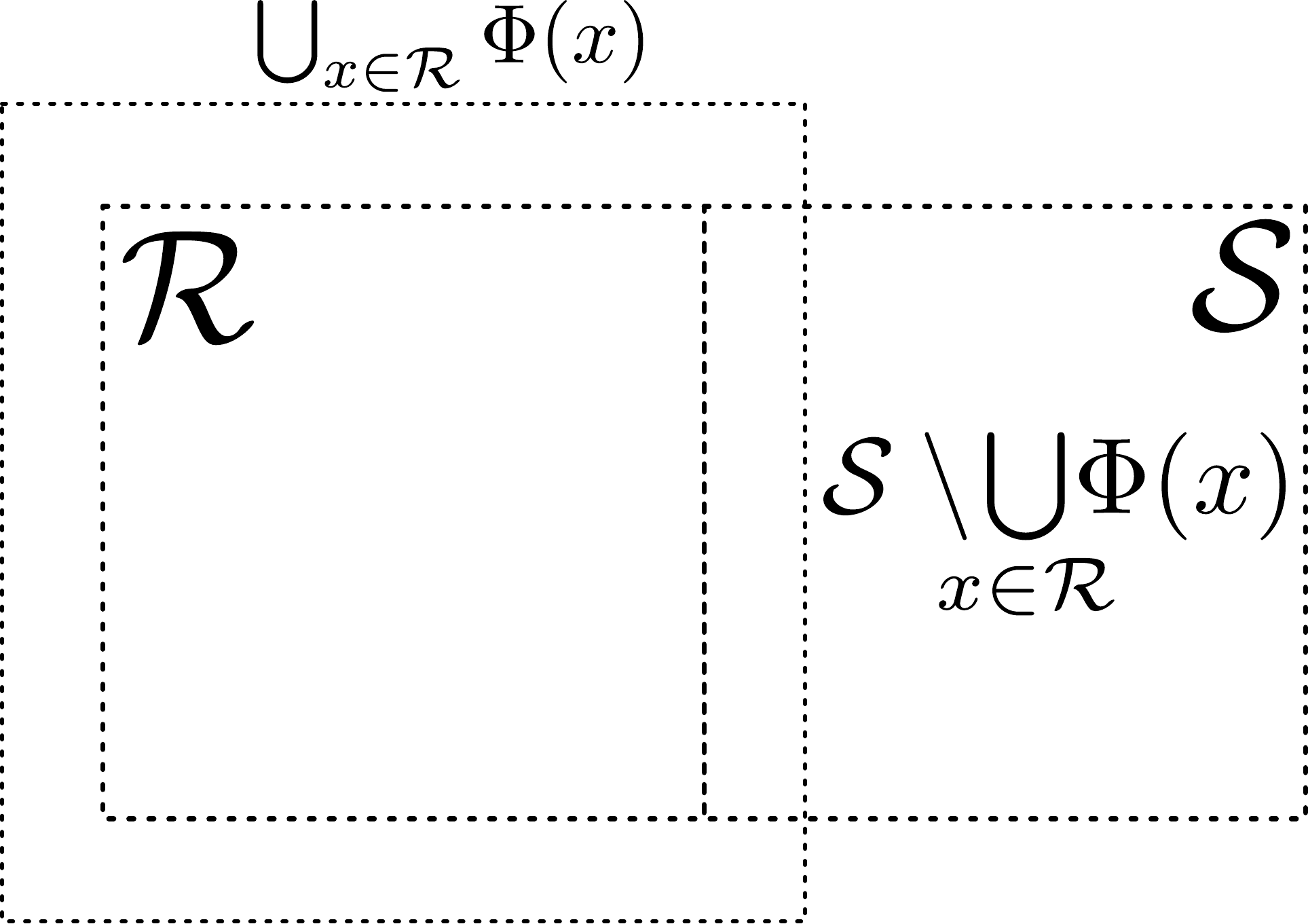}
		\caption{The set $\bigcup_{x\in\reg}\Phi(x)$, and consequently $\sreg \setminus \bigcup_{x\in\reg}\Phi(x)$, does not depend on $x$. The set $\sreg \setminus \bigcup_{x\in\reg}\Phi(x)$ can be partitioned into a finite number (minimum one, here) of hyperrectangles.}
		\label{fig:intersection1}
	\end{subfigure}~~
	\begin{subfigure}[t]{0.32\linewidth}
		\centering
		\includegraphics[height = 1.2in]{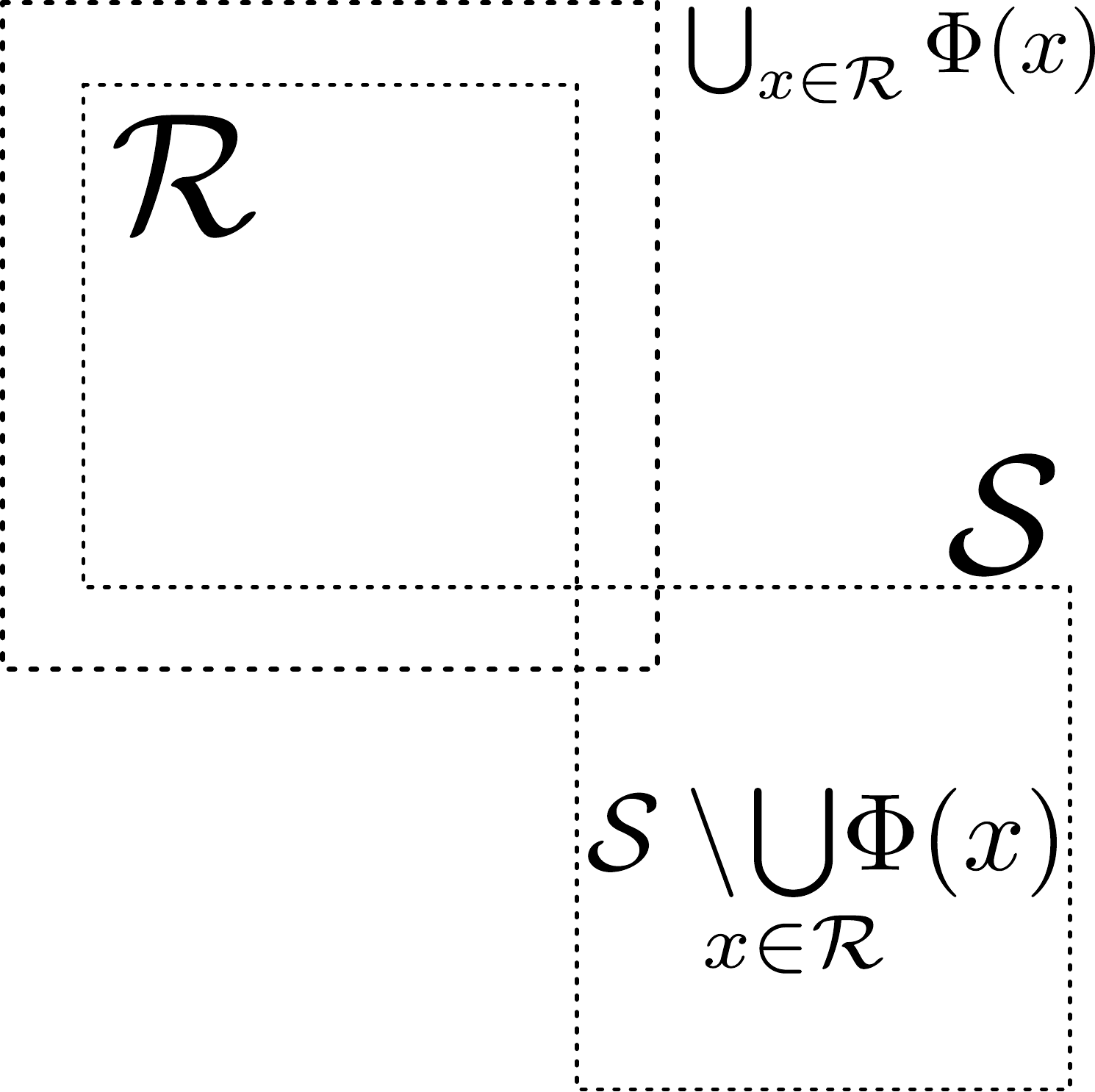}
		\caption{The set $\sreg \setminus \bigcup_{x\in\reg}\Phi(x)$ can be partitioned into a finite number (minimum two, here) of hyperrectangles.}
		\label{fig:intersection2}
	\end{subfigure}
	\caption{The interplay between sets $\reg$, $\sreg$, $\overline{\Phi}(x)\cap \sreg$ and $\sreg \setminus \bigcup_{x\in\reg}\Phi(x)$.}
\end{figure*}

When $s\neq\overline{k}$, the set of integration in \eqref{eq:set-membership1} is neither convex nor linear on $x$ due to $\overline{\Phi}(x)\cap \sreg$ (see Figure \ref{fig:nonconvex_intersection}); thus, we cannot invoke Lemma \ref{lemma:log-concave} and there is no indication that it is straightforward to compute \eqref{eq:main_min_prob}. In this case, we resort to convex relaxations, each of which yield a lower bound on \eqref{eq:main_min_prob} that can be computed easily. These are the following:

\paragraph{Relaxation 1} Notice that $\overline{\Phi}(x)\cap \sreg = \sreg\setminus \Phi(x)$, for any $x\in\reg$. Since $\Phi(x)\subseteq\bigcup_{x\in\reg}\Phi(x)$ for all $x$, it follows that:
\begin{equation*}
	\overline{\Phi}(x)\cap \sreg=\sreg\setminus \Phi(x) \supseteq \sreg \setminus \bigcup_{x\in\reg}\Phi(x)
\end{equation*}
For examples of $\sreg \setminus \bigcup_{x\in\reg}\Phi(x)$ see Figures \ref{fig:intersection1} and \ref{fig:intersection2}. Observe that, since $\bigcup_{x\in\reg}\Phi(x)$ does not depend on $x$, the set $\sreg \setminus \bigcup_{x\in\reg}\Phi(x)$ does not depend on $x$; it is a fixed set, in contrast to $\overline{\Phi}(x)\cap \sreg$. Moreover, since both $\sreg$ and $\bigcup_{x\in\reg}\Phi(x)$ are hyperrectangles, then $\sreg \setminus \bigcup_{x\in\reg}\Phi(x)$ can always be partitioned into a finite number of hyperrectangles $\sreg_1,\dots,\sreg_r$, where $r\leq n$ and $r=0$ in the case where $\sreg \setminus \bigcup_{x\in\reg}\Phi(x)$ is empty. Thus:
\begin{equation}\label{eq:min_relax1}
	\begin{aligned}
		\min\limits_{x\in\reg}\int_{\Phi^{s-1}(x)\times(\overline{\Phi}(x)\cap \sreg)}\hspace{-5mm} \normal(dz|\mu_{\tilde{\zeta}_{s,x}},\Sigma_{\tilde{\zeta}_{s,x}}) \geq&\\ \min\limits_{x\in\reg} \int_{\Phi^{s-1}(x)\times(\sreg \setminus \bigcup_{x\in\reg}\Phi(x))}\hspace{-5mm} \normal(dz|\mu_{\tilde{\zeta}_{s,x}},\Sigma_{\tilde{\zeta}_{s,x}})\geq&\\
		\sum_{i=1}^{r}\min\limits_{x\in\reg}\int_{\Phi^{s-1}(x)\times \sreg_i}\hspace{-5mm} \normal(dz|\mu_{\tilde{\zeta}_{s,x}},\Sigma_{\tilde{\zeta}_{s,x}})&
	\end{aligned}
\end{equation}
The integration sets $\Phi^{s-1}(x)\times \sreg_i$ are convex and linear on $x$. Thus, in the last expression of \eqref{eq:min_relax1} we are dealing with log-concave objective functions, and the $r$ minimization problems attain their minimum at vertices of $\reg$. Hence, we easily solve the $r$ minimization problems to obtain a lower bound on \eqref{eq:main_min_prob}.

\paragraph{Relaxation 2} Here, we employ the law of total probability to write:
\begin{equation*}
	\begin{aligned}
		\pr\Big(\tilde{\zeta}_{s,x}\in\Phi^{s-1}(x)\times(\overline{\Phi}(x)\cap \sreg)\Big) =&\\ \pr\Big(\tilde{\zeta}_{s,x}\in\Phi^{s-1}(x)\times \sreg\Big) - \pr\Big(\tilde{\zeta}_{s,x}\in\Phi^{s-1}(x)\times(\Phi(x)\cap \sreg)\Big)&
	\end{aligned}
\end{equation*}
which gives the following relationship:
\begin{equation}\label{eq:min_relax2}
	\begin{aligned}
		\min\limits_{x\in\reg}\pr\Big(\tilde{\zeta}_{s,x}\in\Phi^{s-1}(x)\times(\overline{\Phi}(x)\cap \sreg)\Big) \geq&\\ \min\limits_{x\in\reg}\pr\Big(\tilde{\zeta}_{s,x}\in\Phi^{s-1}(x)\times \sreg\Big) -&\\ \max\limits_{x\in\reg}\pr\Big(\tilde{\zeta}_{s,x}\in\Phi^{s-1}(x)\times(\Phi(x)\cap \sreg)\Big)&
	\end{aligned}
\end{equation}
The minimization problem in the right-hand side of \eqref{eq:min_relax2} is similar to the ones discussed before (log-concave objective function and hyperrectangle constraint set), and the minimum can be computed easily. However, the set $\Phi(x)\cap \sreg$ not being linear on $x$ makes the maximization problem hard to solve. By employing that $\Phi(x)\cap \sreg\subseteq \sreg\cap\bigcup_{x\in\reg}\Phi(x)$, we relax it by writing:
\begin{align*}
	\max\limits_{x\in\reg}\pr\Big(\tilde{\zeta}_{s,x}\in\Phi^{s-1}(x)\times(\Phi(x)\cap \sreg)\Big) \leq&\\
	\max\limits_{x\in\reg}\pr\bigg(\tilde{\zeta}_{s,x}\in\Phi^{s-1}(x)\times\Big(\sreg\cap\bigcup_{x\in\reg}\Phi(x)\Big)\bigg)
\end{align*}
The set $\sreg\cap\bigcup_{x\in\reg}\Phi(x)$ is a (possibly empty) hyperrectangle and does not depend on $x$; thus, $\Phi^{s-1}(x)\times\Big(\sreg\cap\bigcup_{x\in\reg}\Phi(x)\Big)$ is convex and linear on $x$. Hence, the maximization problem in the right-hand side of the above equation is a convex program (log-concave objective function over the convex constraint set $\reg$), and can be easily solved via regular convex optimization techniques. By computing the exact minimum in the right-hand side of \eqref{eq:min_relax2} and an upper bound on the maximum-term as discussed here, we obtain a lower bound on \eqref{eq:main_min_prob}.

\paragraph{Relaxation 3} Continuing from \eqref{eq:min_relax2}, we propose a different relaxation for the maximization problem in the right-hand side of \eqref{eq:min_relax2}. Specifically, by employing Bayes's rule:
\begin{equation}\label{eq:min_relax_3}
	\begin{aligned}
		\max\limits_{x\in\reg}\pr\Big(\tilde{\zeta}_{s,x}\in\Phi^{s-1}(x)\times(\Phi(x)\cap \sreg)\Big) \leq&\\ \max\limits_{x\in\reg}\pr\Big(\tilde{\zeta}_{s,x}\in\Phi^{s}(x)| \zeta_{s,x}\in \sreg\Big) \cdot \max\limits_{x\in\reg}\pr(\zeta_{s,x}\in \sreg)&
	\end{aligned}
\end{equation}
The term $\max\limits_{x\in\reg}\pr(\zeta_{s,x}\in \sreg)$ can be computed exactly easily, as $\pr(\zeta_{s,x}\in \sreg)$ is log-concave on $x$. For the term $\max\limits_{x\in\reg}\pr\Big(\tilde{\zeta}_{s,x}\in\Phi^{s}(x)| \zeta_{s,x}\in \sreg\Big)$, we make use of the following bound:
\begin{proposition}\label{prop:luca_inequality}
	The following holds:
	\begin{equation}\label{eq:luca_inequality}
		\begin{aligned}
			\max\limits_{x\in\reg}\pr\Big(\tilde{\zeta}_{s,x}\in\Phi^{s}(x)| \zeta_{s,x}\in \sreg\Big) \leq&\\ \max\limits_{(x,v)\in\reg\times \sreg} \pr\Big(\tilde{\zeta}_{s,x}\in\Phi^{s}(x)| \zeta_{s,x}=v\Big)&
		\end{aligned}
	\end{equation}
\end{proposition}
\begin{proof}
	The proof is the same as in \cite[Lemma 2]{blaas2019adversarial}.
\end{proof}
To compute the right-hand side of \eqref{eq:luca_inequality}, we use the fact that the random variable $\xi=(\tzeta_{s,x}|\zeta_{l,x}=v)$ is normally distributed:
\begin{corollary}[to Proposition \ref{prop:normal}]\label{cor:conditional_normal}
	Consider the random variable $\xi=(\tzeta_{s,x}|\zeta_{l,x}=v)$, where $l\in\mathbb{N}_{[0,s]}$, and $v\in\real^{n}$. Then $\xi\sim\normal(\mu_\xi(x,v), \Sigma_\xi)$, where:
	\begin{align*}
		&\mu_\xi(x,v) = \expec(\tzeta_{s,x})-\Sigma_{\tzeta_{s,x},\zeta_{l,x}}\Sigma_{\zeta_{l,x}}^{-1}(v-\expec(\zeta_{l,x}))\\
		&\Sigma_\xi = \Sigma_{\tzeta_{s,x}} - \Sigma_{\tzeta_{s,x},\zeta_{l,x}}\Sigma_{\zeta_{l,x}}^{-1}\Sigma_{\zeta_{l,x},\tzeta_{s,x}},
	\end{align*}
	where $\Sigma_{\zeta_{l,x}} = \cov(l,l)$, $\Sigma_{\tilde{\zeta}_{s,x}}$, $\expec(\tilde{\zeta}_{s,x})$ and $\expec(\zeta_{l,x})$ are obtained from Proposition \ref{prop:normal}, and $\Sigma_{\zeta_{l,x},\tzeta_{s,x}} =\Sigma_{\tzeta_{s,x},\zeta_{l,x}}^\top= \begin{bmatrix}
		\cov(l,1) &\cov(l,2) &\dots &\cov(l,s)
	\end{bmatrix}$.
\end{corollary}
\begin{proof}
	Straightforward application of the well-known formula for conditional normal distributions \cite{eaton1983multivariate}.
\end{proof}
Thus, we have that:
\begin{equation*}
	\begin{aligned}
		\max\limits_{(x,v)\in\reg\times \sreg} \pr\Big(\tilde{\zeta}_{s,x}\in\Phi^{s}(x)| \zeta_{s,x}=v\Big) =&\\ \max\limits_{(x,v)\in\reg\times \sreg}\int_{\Phi^s(x)}\normal(dz|\mu_\xi(x,v), \Sigma_\xi)&
	\end{aligned}
\end{equation*}
Observe that $\mu_\xi(x,v)$ is affine on the optimization variables $(x,v)$, and $\Phi^s(x)$ is obviously convex and linear on $x$. Thus, the objective function of the above maximization problem is log-concave. Finally, since the set of constraints $\reg\times \sreg$ is convex, we deduce that computing the right-hand side of \eqref{eq:luca_inequality} is a convex program. Combining \eqref{eq:luca_inequality} with \eqref{eq:min_relax_3} and \eqref{eq:min_relax2} yields an easily computable bound on \eqref{eq:main_min_prob}.

\begin{remark}
	Gaussian integrals over hyperrectangles are often encountered in fields such as statistics and learning, and many algorithms exist for their numerical computation (e.g., Genz's algorithm \cite{genz1992numerical} or python's \emph{scipy.stats.multivariate\_normal} \cite{mvnun}). 
\end{remark}

\subsection{Upper Bounds on Transition Probabilities}
We proceed to computing upper bounds on:
\begin{equation} \label{eq:main_max_prob}
	\max_{x\in\reg}\pr(\zeta(s;x)\in \sreg,\tau(x)=s)
\end{equation}
Again, the case where $s=\overline{k}$ is easy: it corresponds to a convex program, and \eqref{eq:main_max_prob} is computed exactly. For the case where $s\neq\overline{k}$, we employ a relaxation similar to Relaxation 3 described in the previous. In particular, as in \eqref{eq:min_relax2}, we write:
\begin{equation}\label{eq:max_relax}
	\begin{aligned}
		\max\limits_{x\in\reg}\pr\Big(\tilde{\zeta}_{s,x}\in\Phi^{s-1}(x)\times(\overline{\Phi}(x)\cap \sreg)\Big) \leq&\\ \max\limits_{x\in\reg}\pr\Big(\tilde{\zeta}_{s,x}\in\Phi^{s-1}(x)\times \sreg\Big) -&\\ \min\limits_{x\in\reg}\pr\Big(\tilde{\zeta}_{s,x}\in\Phi^{s-1}(x)\times(\Phi(x)\cap \sreg)\Big)&
	\end{aligned}
\end{equation}
The term $\max\limits_{x\in\reg}\pr\Big(\tilde{\zeta}_{s,x}\in\Phi^{s-1}(x)\times \sreg\Big)$ is computed easily, through convex optimization. For the other term in the right-hand side of \eqref{eq:max_relax}, we write as in \eqref{eq:min_relax_3}:
\begin{equation}\label{eq:max_relax2}
	\begin{aligned}
		\min\limits_{x\in\reg}\pr\Big(\tilde{\zeta}_{s,x}\in\Phi^{s-1}(x)\times(\Phi(x)\cap \sreg)\Big) \geq&\\ \min\limits_{x\in\reg}\pr\Big(\tilde{\zeta}_{s,x}\in\Phi^{s}(x)| \zeta_{s,x}\in \sreg\Big) \cdot \min\limits_{x\in\reg}\pr(\zeta_{s,x}\in \sreg)&
	\end{aligned}
\end{equation}
Given the discussion of the previous section, it is clear that: a) $\min\limits_{x\in\reg}\pr(\zeta_{s,x}\in \sreg)$ is computed exactly (by traversing the vertices of $\reg$), and b) a lower bound on $\min\limits_{x\in\reg}\pr\Big(\tilde{\zeta}_{s,x}\in\Phi^{s}(x)| \zeta_{s,x}\in \sreg\Big)$ is computed by employing Proposition \ref{prop:luca_inequality} and Corollary \ref{cor:conditional_normal}, which yield log-concave minimization over the polytope $\reg\times \sreg$.

\subsection{Transitions to $\uns$}\label{sec:transitions_to_uns}
According to the last two inequalities in \eqref{eq:probability_bounds}, for transitions to $\uns$ we are interested in:
\begin{equation*}
	(\min_{x\in\reg}\text{ or})\quad \max_{x\in\reg}\pr(\zeta(s;x)\in\overline{X},\tau(x)=s)
\end{equation*}
We focus on the maximization, as minimization follows identical steps. By the law of total probability, we have:
\begin{equation}
	\begin{aligned}
		\max_{x\in\reg}\pr(\zeta(s;x)\in\overline{X},\tau(x)=s) \leq&\\ \max_{x\in\reg}\pr(\tau(x)=s) - \min_{x\in\reg}\pr(\zeta(s;x)\in X,\tau(x)=s)&
	\end{aligned}
\end{equation}
Note that, since $X$ is a hyperrectangle, the term $\min_{x\in\reg}\pr(\zeta(s;x)\in X,\tau(x)=s)$ can be treated exactly as discussed in the previous sections (where $X$ takes the place of $\sreg$). Regarding $\pr(\tau(x)=s)$, we have the following two cases:

\paragraph{$s = \overline{k}$} In this case:
\begin{equation*}
	\pr(\tau(x)=s) = \pr(\tzeta_{s,x}\in\Phi^{s-1}(x))
\end{equation*}
Thus, $\max_{x\in\reg}\pr(\tau(x)=s) = \max_{x\in\reg} \pr(\tzeta_{s-1,x}\in\Phi^{s-1}(x))$, which can be computed easily (log-concave objective function and hyperrectangular constraint set).

\paragraph{$s \neq \overline{k}$} In this case, by the law of total probability:
\begin{equation*}
	\begin{aligned}
		\pr(\tau(x)=s) = \pr(\tzeta_{s,x}\in\Phi^{s-1}(x)\times\overline{\Phi}(x))=&\\
		\pr(\tzeta_{s-1,x}\in\Phi^{s-1}(x)) - \pr(\tzeta_{s,x}\in\Phi^{s}(x))
	\end{aligned}
\end{equation*} 
where when $s=1$ we have abusively denoted $\tzeta_{0,x} = x$ and $\Phi^{0}(x)=x$. Thus, we have:
\begin{equation*}
	\begin{aligned}
		\max\limits_{x\in\reg}\pr(\tau(x)=s) \leq&\\
		\max\limits_{x\in\reg}\pr(\tzeta_{s-1,x}\in\Phi^{s-1}(x)) - \min\limits_{x\in\reg}\pr(\tzeta_{s,x}\in\Phi^{s}(x))
	\end{aligned}
\end{equation*}
and both terms in the right-hand side can be computed easily as discussed in the previous sections (log-concave objective functions and hyperrectangular constraint sets).

\section{Numerical Examples} \label{sec:examples}
We, now, demonstrate our theoretical results with a numerical example. 
Consider a stochastic PETC system \eqref{snh}-\eqref{trig_cond} with:
\begin{equation*}
	A = \begin{bmatrix}
		-4 &3\\ -2 &1
	\end{bmatrix}, B = \begin{bmatrix}
		1\\0\end{bmatrix}, 
	K = \begin{bmatrix}
		-2 &3
	\end{bmatrix}, B_w = \begin{bmatrix}
		2.5 &0 \\0 &2.5
	\end{bmatrix}
\end{equation*}
and $\epsilon = 0.25$, $h = 0.006$, $\overline{k} = 3$. We are interested in assessing the sampling behaviour of the system for initial conditions in $X = [-1.2,1.2]^2$. Following Remark \ref{rem:partition_Y}, we partition $Y=[-2,2]^2$ into 2500 equal rectangles, and construct the IMC as described in the previous.

First, consider the multiplicative reward from Example 3 in Section \ref{sec:sampling_behaviour} and a horizon $N=5$. Recall that, in this case, the expected reward expresses the probability that there is no intersampling time $s=\overline{k}$ in the first 5 triggers. 
As dictated by Theorem \ref{main_theorem}, we equip the IMC with rewards $\underline{R},\overline{R}$, which are as follows for any $q\in Q_{\imc}$:
\begin{align*}
	&\underline{R}(q) = \left\{\begin{aligned}
		&0,\quad \text{if }q=\uns \text{ or } \proj_{\mathbb{N}}(q)=\overline{k}\\
		&1, \quad \text{otherwise}
	\end{aligned}\right.\\
	&\overline{R}(q) = \left\{\begin{aligned}
		&0,\quad \text{if } \proj_{\mathbb{N}}(q)=\overline{k}\\
		&1, \quad \text{otherwise}
	\end{aligned}\right.
\end{align*}
For all $q_0\in Q_{\imc}\setminus\uns$, we calculate $\inf_{\adv\in\Pi}\hspace{-.5mm}\expec^{q_0}_\adv[\underline{g}_{\mul,N}(\tilde{\omega})]$ and $\sup_{\adv\in\Pi}\hspace{-.5mm}\expec^{q_0}_\adv[\overline{g}_{\mul,N}(\tilde{\omega})]$, by employing the value iteration introduced in \eqref{eq:viimc_mul}. The adversary that gives rise to each bound is the so-called \emph{o-maximizing MDP} and can be found easily (see \cite{givan_bmdps} and \cite{lahijanian2015dt_imcs}).
\begin{figure}[h!]
	\centering
	\includegraphics[width = 3in]{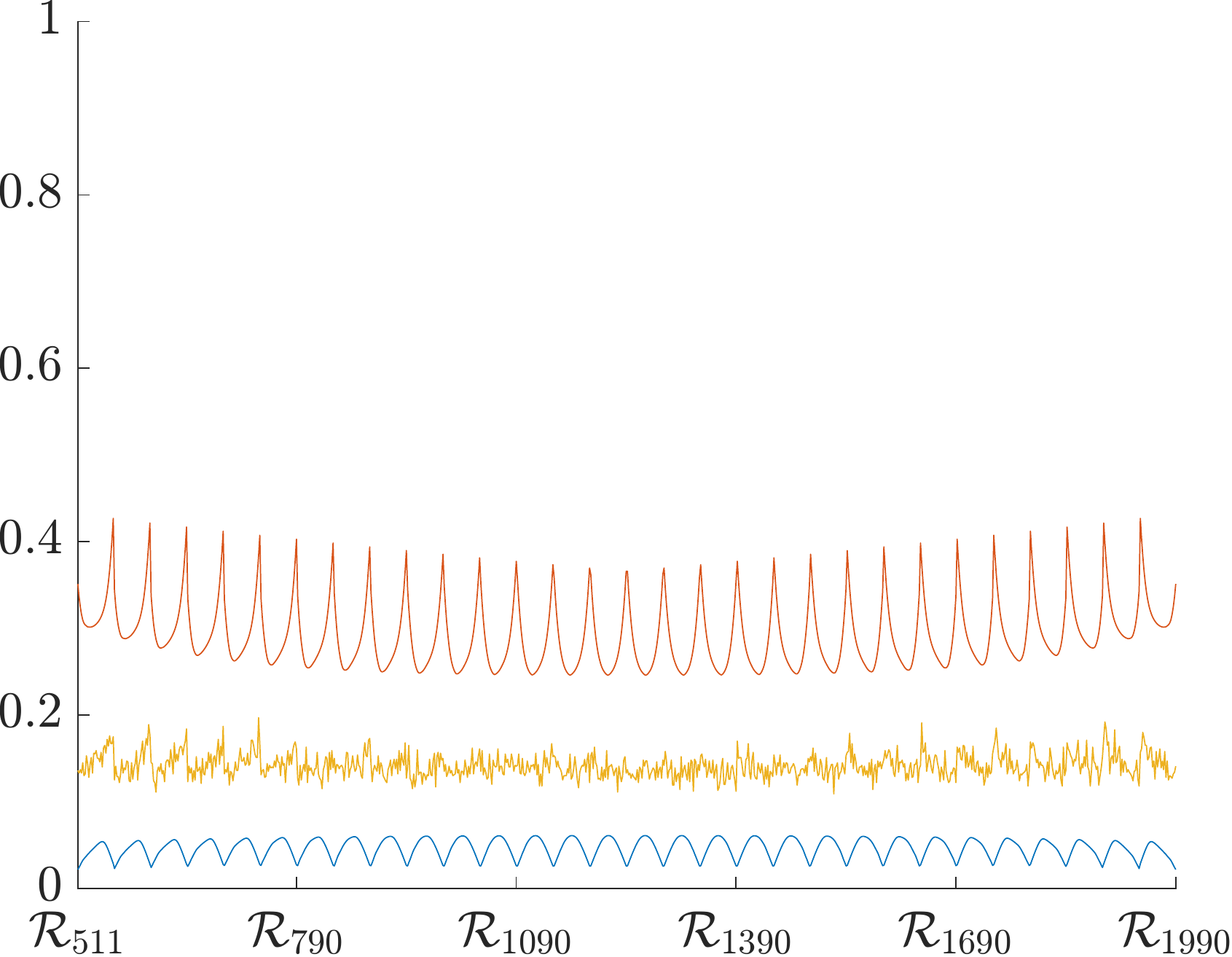}
	\caption{The blue and red lines are the computed lower and upper bounds, respectively, on the expected multiplicative reward from Example 3 in Section \ref{sec:sampling_behaviour} starting from any initial condition $x_0\in \reg_i$ (initial intersampling time is assumed $s_0=0$), for all regions $\reg_i\subset[-1.2,1.2]^2$ in the partition. The yellow (middle) line is the statistical estimate of the expected reward for a random initial condition from each region.}
	\label{fig:not_sample_inside}
\end{figure}

The obtained bounds for all $q_0=(\reg,0)\in Q_\reg\times\{0\}$, with $\reg\subset[-1.2,1.2]^2$, are shown in Figure \ref{fig:not_sample_inside}. We only consider the case where the initial intersampling time $s_0=0$, as commented in Remark \ref{rem:init_cond}\footnote{This is with no loss to generality, as $s_0$ does not affect the evolution of the system: for different $s_0$ and the same realization of the Wiener process, the sample path evolves exactly the same.}. From the obtained bounds, one can expect from the system a high probability of sampling with intersampling time $\overline{k}$. Thus, based on that observation, an engineer who is to implement the PETC system, could decide to further increase the maximum allowed intersampling time, in order to allow the system to sample even less frequently.
\begin{figure}[h!]
	\centering
	\includegraphics[width = 3in]{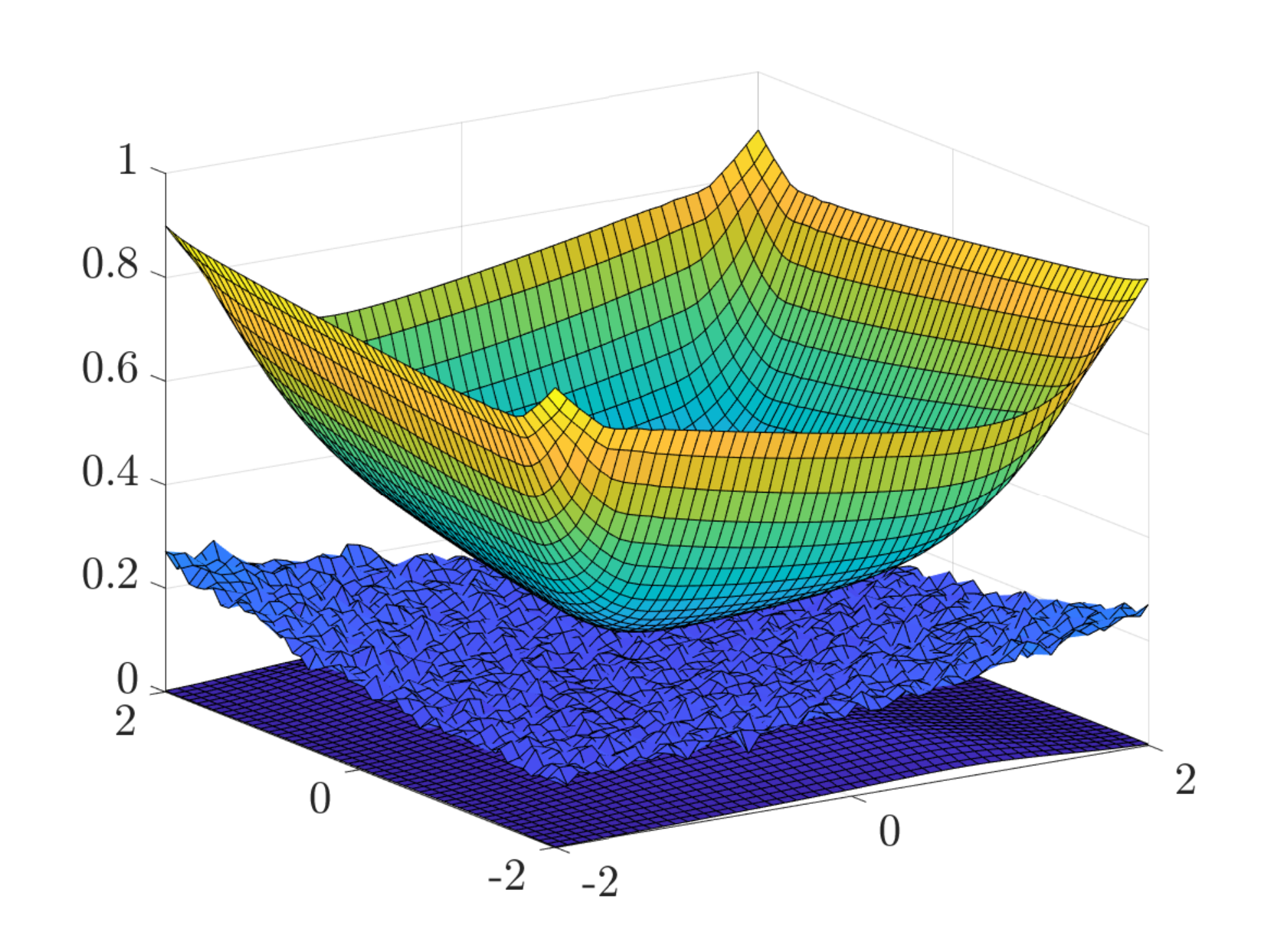}
	\caption{Surface plot of the obtained lower and upper bounds on the expected multiplicative reward from Example 3 in Section \ref{sec:sampling_behaviour} for all regions $\reg_i\subset[-2,2]^2$ in the partition (x-axis). The surface on the bottom is the lower bound, the surface at the top is the upper bound, and the one in the middle is the statistical estimate of the expected reward for a random initial condition from each region, as obtained from simulations.}
	\label{fig:surf_not_sample}
\end{figure}

Figure \ref{fig:not_sample_inside}, also, shows the statistical estimate of the expected reward, as derived by simulations. Specifically, for all $q_0\in Q_{\reg}\times\{0\}$ with $\reg\subset[-1.2,1.2]^2$, we pick a random initial condition $y_0\in q_0$ and simulate 1000 sample paths, with a horizon of 5 triggers (the simulation stops after the 5th trigger). Each sample path that does not generate any intersampling time $s=\overline{k}$ is counted, and the total count is divided by 1000 to obtain a statistical estimate of the true probability. Figure \ref{fig:not_sample_inside} shows that, as expected by Theorem \ref{main_theorem}, the statistical estimate is confined within the computed bounds. Finally, Figure \ref{fig:surf_not_sample} is a surface plot illustrating the obtained bounds and the statistical estimate for all regions $\reg_i\in[-2,2]^2$, supporting what is discussed in Remark \ref{rem:partition_Y}: regions closer to the boundary of the partition correspond to more conservative bounds. 

\begin{figure}[h!]
	\centering
	\includegraphics[width = 3in]{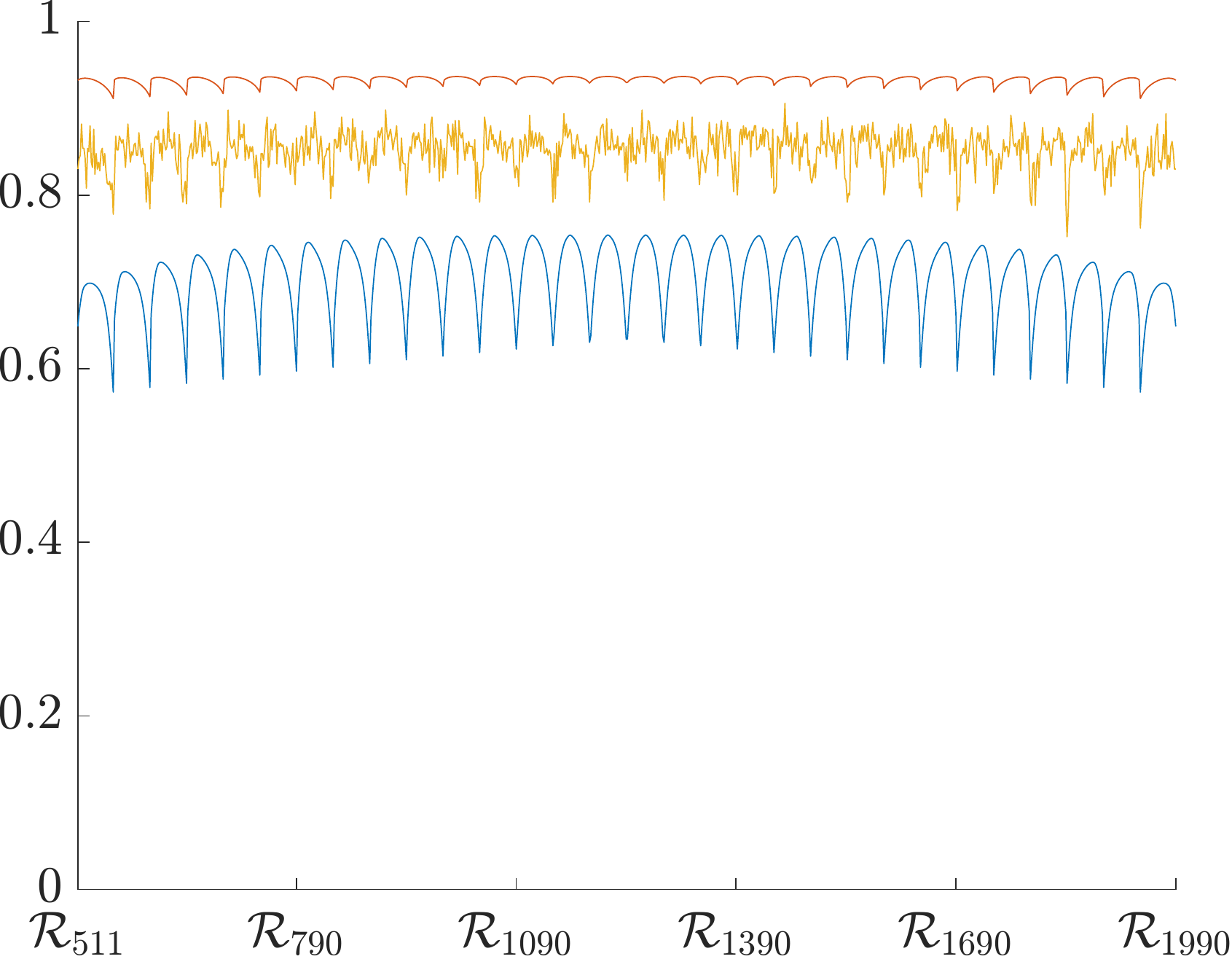}
	\caption{The blue and red lines are lower and upper bounds, respectively, on the bounded-until probability \eqref{eq:example_until} starting from any initial condition $x_0\in \reg_i$ (initial intersampling time is assumed $s_0=0$), for all regions $\reg_i\subset[-1.2,1.2]^2$. The yellow line (the one in the middle) is the statistical estimate of the probability for a random initial condition from each region.}
	\label{fig:sample_while_inside}
\end{figure} 
Next, to demonstrate our results' extension to PCTL, we derive bounds on the following bounded-until probability:
\begin{equation}\label{eq:example_until}
	\begin{aligned}
		\pr^{y_0}_{\Y_N}\Big(\exists i\in\mathbb{N}_{[0,5]} \text{ s.t. }&\proj_{\mathbb{N}}(\omega(i))=\overline{k} \text{ and }\forall k\leq i, \text{ }\omega(k)\notin\uns\Big)
	\end{aligned}
\end{equation} 
This is the probability that the state stays in $Y$ \emph{until} there is a trigger $s=\overline{k}$, in a horizon $N=5$. Figure \ref{fig:sample_while_inside} shows the results.

Finally, for completeness, we calculate bounds on the expected average intersampling time for $N=5$, as introduced in Example 1, Section \ref{sec:sampling_behaviour}. Since we assume $s_0=0$, which implies that we are only interested in the average of the 5 subsequent triggers, we use $N$ in the denominator, instead of $N+1$. The results are illustrated in Figure \ref{fig:avg_inside}. The obtained bounds could be used to compare the average sampling performance of this particular PETC design with some other implementation; e.g. it is evident that, on average, it samples considerably more efficiently than a periodic implementation with period $h$. Alternatively, they could be used to forecast the expected average occupation of the communication channel. 
\begin{figure}
	\centering
	\includegraphics[width = 3in]{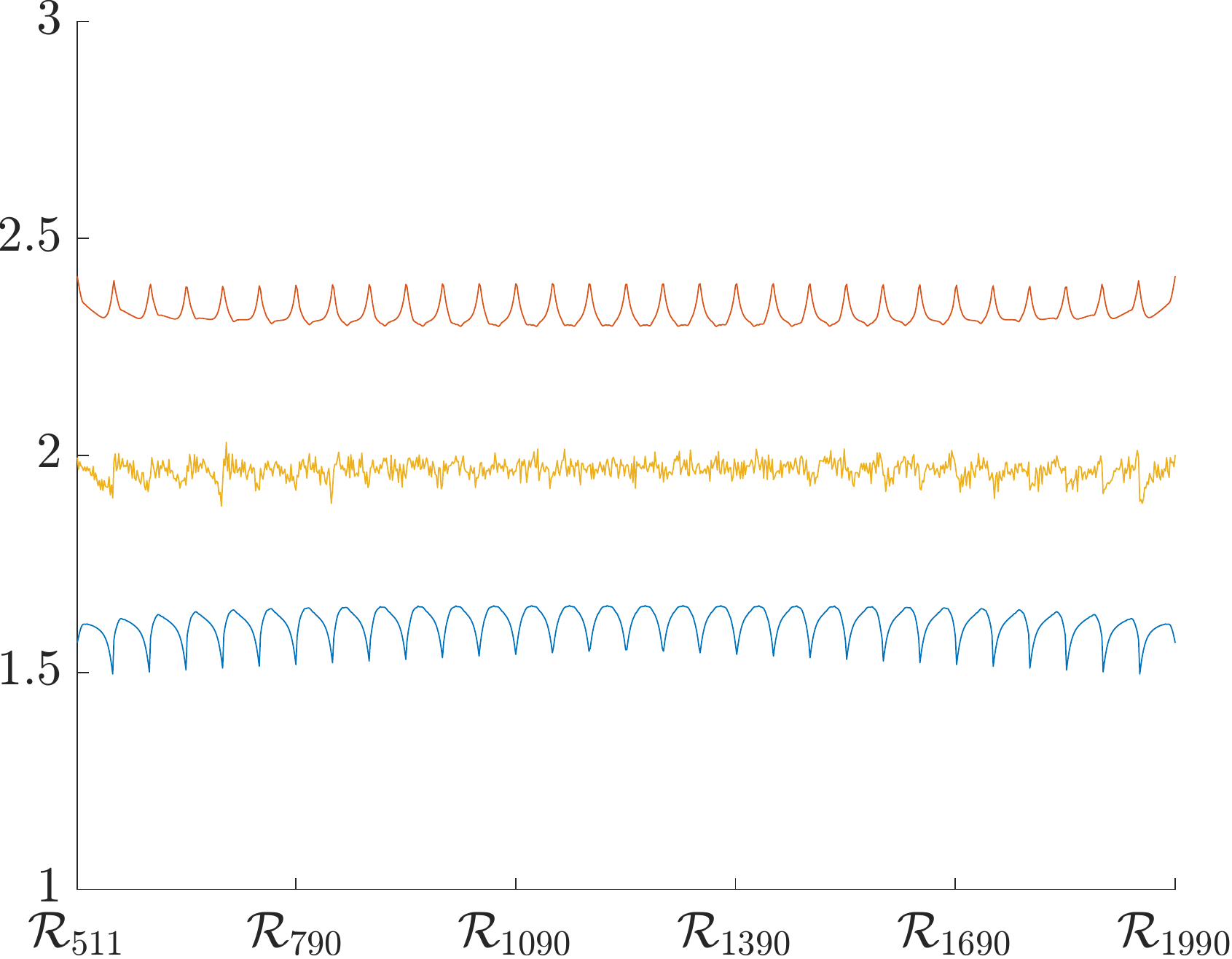}
	\caption{The blue and red lines are lower and upper bounds on the expected average intersampling time starting from any initial condition $x_0\in \reg_i$ (initial intersampling time is assumed $s_0=0$), for regions $\reg_i\subset[-1.2,1.2]^2$. The yellow line (the one in the middle) is the statistical estimate of the expected average for a random initial condition from each region.}
	\label{fig:avg_inside}
\end{figure}

\section{Conclusion}
In this work, we have computed bounds on metrics associated to the sampling behaviour of linear stochastic PETC systems, by constructing IMCs abstracting the sampling behaviour and equipping them with suitable rewards. The metrics are expectations of functions of sequences of intersampling times and state measurements, that take the form of cumulative, average or multiplicative rewards. Numerical examples have been provided to demonstrate the effectiveness of the proposed framework in practice. Specifically, for a given system, we have computed the expected average intersampling time and the probability of triggering with the maximum allowed intersampling time, in a finite horizon Moreover, we have computed bounds on a bounded-until probability, demonstrating extensibility of our approach to PCTL properties. Overall, the framework presented here, enables the formal study  of PETC's sampling behaviour and the assessment of its sampling (vs. control) performance.

Future work will focus on the following: a) extending the class of systems considered, b) investigating how the state-space partition proposed in \cite{gleizer2020scalable} for deterministic linear systems can be employed here, for better and more scalable results, and c) endowing the IMCs with actions, which will allow for scheduling ETC data traffic in networks shared by multiple ETC loops, such that performance criteria are met and optimized (e.g. minimizing packet collisions).

\section{Acknowledgements}
The authors thank Daniel Jarne Ornia and Gabriel de Albuquerque Gleizer for helpful discussions on this work.
\section{Appendix}

\subsection{Technical Lemmas and Proof of Theorem \ref{main_theorem}}\label{app:main_theorem}
In this subsection, we first provide some technical lemmas, and then prove Theorem \ref{main_theorem}. Let us introduce some notation and terminology. We constrain ourselves to Markovian adversaries. The value of such adversaries depends only on the time-step $i$ and the given state $q\in Q_{\imc}$, i.e. $\adv(i,q) = p_{i,q}\in\Gamma_q$. From now on, we abusively write $\adv(i,q,q')=p_{i,q}(q')$, for any $q'\in Q_\imc$, to denote the transition probability from $q$ to $q'$ at time $i$, under adversary $\adv$. 
Moreover, for $s_i,s_{i+1}\in\mathbb{N}_{[0,\overline{k}]}$, $x_i\in\real^{n_\zeta}$, $X_{i+1}\subseteq\real^{n_\zeta}$, denote:
\begin{equation}\label{eq:kernel}
	\begin{aligned}
		&T((X_{i+1},s_{i+1})|(x_i,s_i)):=\pr^{y_0}_{\Y_N}(\omega(i+1)\in (X_{i+1},s_{i+1})|\omega(i)=(x_i,s_i))
	\end{aligned}
\end{equation}
This notation is common in the literature of stochastic systems and $T$ is often called \textit{transition kernel}. Let us abuse notation and write $\int_QT(dy'|y)$, for some $y\in Q$, to denote $\sum_{s'\in\mathbb{N}_{[0,\overline{k}]}}\int_{\real^n}T((dx',s')|y)$.

We proceed to stating the technical lemmas. The first one provides a relationship indicating that the expected cumulative reward can be written as a \emph{value function} defined via \emph{value iteration}, which is a trivial extension of the value iteration in \cite{givan_bmdps} to finite horizons and time-varying adversaries. The second and third lemmas provide some useful bounds, which are employed in the proof of Theorem \ref{main_theorem}. 
\begin{lemma}\label{lemma:vi}
	Given IMC $\sys_{\imc}$ from \eqref{eq:our_imc}, equipped with a reward function $\underline{R}:Q_\imc \to \underline{R}_\max$, any Markovian adversary $\adv\in\Pi$ and any $q_0 \in Q_\imc$, we have that:
	\begin{equation}\label{eq:lemma_vi1}
		\expec_\adv[\sum_{j=i}^N\gamma^{j-i}\underline{R}(\tilde{\omega}(j))|\tilde{\omega}(i)=q_0] = \underline{V}_{\adv,i}(q_0)
	\end{equation}
	where for all $q\in Q_\imc$ and $i\in\mathbb{N}_{[0,N-1]}$:
	\begin{subequations}\label{eq:viimc}
		\begin{align}
			&\underline{V}_{\adv,N}(q) = \underline{R}(q)\\
			&\underline{V}_{\adv,i}(q) = \underline{R}(q) + \gamma\hspace{-2mm}\sum_{q'\in Q_\imc}\hspace{-3mm}\underline{V}_{\adv,i+1}(q')\adv(i, q,q')
		\end{align}
	\end{subequations}
	Similarly, for all $y_0\in Q$:
	\begin{equation}\label{eq:lemma_vi2}
		\expec_{\pr^{y_0}_\Y}[\sum_{j=i}^N\gamma^{j-i}R(\omega(j))|\omega(i)=y_0] = V_{i}(y_0)
	\end{equation}
	where for all $y\in Q$ and $i\in\mathbb{N}_{[0,N-1]}$:
	\begin{subequations}\label{eq:vi_system}
		\begin{align}
			&V_{N}(y) = R(y)\\
			&V_{i}(y) = R(y) + \gamma\int_{Q}V_{i+1}(y')T(dy'|y)
		\end{align}
	\end{subequations} 
	Consequently, we have:
	\begin{equation}\label{eq:lemma_vi3}
		\begin{aligned}
			\expec^{q_0}_\adv[\underline{g}_{\cum,N}(\tilde{\omega})] &= \expec_\adv[\sum_{j=0}^N\gamma^j\underline{R}(\tilde{\omega}(j))|\tilde{\omega}(0)=q_0] \\&= \underline{V}_{\adv,0}(q_0)\\ 
			\expec_{\pr^{y_0}_\Y}[g_{\cum,N}(\omega)] &= \expec_{\pr^{y_0}_\Y}[\sum_{j=0}^N\gamma^jR(\omega(j))|\omega(0)=y_0]\\&= V_{0}(y_0)
		\end{aligned}
	\end{equation}
\end{lemma}
\begin{proof}
	We prove \eqref{eq:lemma_vi1} by induction. The proof of \eqref{eq:lemma_vi2} is identical, and then \eqref{eq:lemma_vi3} follows immediately. It obviously holds that $\underline{V}_{\adv,N}(q_0) = R(q_0) = \expec_\adv[\sum_{j=N}^N\gamma^{j-N}\underline{R}(\tilde{\omega}(j))|\tilde{\omega}(N)=q_0]$ for all $q_0 \in Q_\imc$. Now, assume that \eqref{eq:lemma_vi1} holds for some $i\in\mathbb_{N}_{[1,N]}$. Then:
	\begin{align*}
		\underline{V}_{\adv,i-1}(q_0) =&\\ \underline{R}(q_0) + \gamma\hspace{-2mm}\sum_{q'\in Q_\imc}\hspace{-3mm}\underline{V}_{\adv,i}(q')\adv(i-1, q_0,q')=&\\
		\underline{R}(q_0) + \gamma\hspace{-2mm}\sum_{q'\in Q_\imc}\hspace{-3mm}\expec_\adv\Big[\sum_{j=i}^N\gamma^{j-i}\underline{R}(\tilde{\omega}(j))|\tilde{\omega}(i)=q'\Big]\adv(i-1, q_0,q')=&\\
		\underline{R}(q_0) + \expec_\adv\Big[\sum_{j=i}^N\gamma^{j-i+1}\underline{R}(\tilde{\omega}(j))|\tilde{\omega}(i-1)=q_0\Big]=&\\
		\expec_\adv\Big[\underline{R}(q_0) + \sum_{j=i}^N\gamma^{j-i+1}\underline{R}(\tilde{\omega}(j))|\tilde{\omega}(i-1)=q_0\Big] =&\\
		\expec_\adv\Big[\underline{R}(q_0) + \gamma\underline{R}(\tilde{\omega}(i)) + \gamma^2\underline{R}(\tilde{\omega}(i+1))+\dots|\tilde{\omega}(i-1)=q_0\Big] =&\\
		\expec_\adv[\sum_{j=i-1}^N\gamma^{j-i+1}\underline{R}(\tilde{\omega}(j))|\tilde{\omega}(i-1)=q_0]&
	\end{align*}
	where:
	\begin{itemize}
		\item in the second equality we used the induction assumption that we made;
		\item in the third equality we put $\gamma$ inside the expectation, and we used the law of total expectation;
		\item and in the fourth equality we put $\underline{R}(q_0)$ inside the expectation.
	\end{itemize}
	Thus \eqref{eq:lemma_vi1} is proven by induction, and the proof is completed.
\end{proof}

\begin{lemma}\label{lemma_uns}
	Given any adversary $\adv\in\Pi$, for all $y\in\real^{n_\zeta}\times\mathbb{N}_{[1,\overline{k}]}$ and for all $i\in\mathbb{N}_{[1,N]}$:
	\begin{equation}\label{eq:lemma_uns_eq}
		\underline{V}_{\adv,i}(\uns)\leq V_{i}(y)
	\end{equation}
\end{lemma}
\begin{proof}
	From Lemma \ref{lemma:vi}, we know that:
	\begin{equation}\label{eq:lemma_uns1}
		\begin{aligned}
			V_{i}(y) &= \expec_{\pr^{y}_\Y}[\sum_{j=i}^N\gamma^{j-i}R(\omega(j))|\omega(i)=y]\\
			\underline{V}_{\adv,i}(\uns) &= \expec_\adv[\sum_{j=i}^N\gamma^{j-i}\underline{R}(\tilde{\omega}(j))|\tilde{\omega}(i)=\uns]=\sum_{j=i}^N\gamma^{j-i}\underline{R}(\uns) = \sum_{j=i}^N\gamma^{j-i}\min\limits_{(x,s)\in \real^{n_\zeta}\times\mathbb{N}_{[1,\overline{k}]}}\hspace{-7mm}R(x,s)
		\end{aligned}
	\end{equation}
	where in the second equation, the second equality comes from the fact that $\uns$ is absorbing for any $\adv\in\Pi$ and the third equality comes from \eqref{eq:underline_R}. 
	
	From Remark \ref{rem:zero_tau}, since $y\in\real^{n_\zeta}\times\mathbb{N}_{[1,\overline{k}]}$ we can deduce that for all $j\in\mathbb{N}_{[i,N]}$ we have:
	\begin{equation*}
		\pr_{\Y_N}(\omega(j)\in\real^{n_\zeta}\times\mathbb{N}_{[1,\overline{k}]}|\omega(i)=y)=1
	\end{equation*}
	Thus, we have:
	\begin{equation}
		\min\limits_{(x,s)\in \real^{n_\zeta}\times\mathbb{N}_{[1,\overline{k}]}}\hspace{-7mm}R(x,s) \leq R(\omega(j)), \text{ } \text{for all }j\in\mathbb{N}_{[i,N]} \text{ a.s.}
	\end{equation}
	where a.s. means ``almost surely". By combining the above equation with \eqref{eq:lemma_uns1}, equation \eqref{eq:lemma_uns_eq} follows.
\end{proof}

\begin{lemma}\label{lemma:uns-int}
	Given any adversary $\adv\in\Pi$ and any $y\in Q$:
	\begin{equation}\label{eq:lemma_uns_int}
		\underline{V}_{\adv,i}(\uns)\int_{\uns}T(dy'|y)\leq \int_{\uns}V_{i}(y')T(dy'|y)
	\end{equation}
\end{lemma}
\begin{proof}
	We have that:
	\begin{align*}
		\int_{\uns}V_{i}(y')T(dy'|y)=&\\
		\int_{\overline{X}\times\mathbb{N}_{[1,\overline{k}]}}V_{i}(y')T(dy'|y)+\cancel{\int_{\overline{X}\times\{0\}}V_{i}(y')T(dy'|y)}^0\geq&\\
		\underline{V}_{\adv,i}(\uns)\int_{\overline{X}\times\mathbb{N}_{[1,\overline{k}]}}T(dy'|y)=&\\
		\underline{V}_{\adv,i}(\uns)\int_{\overline{X}\times\mathbb{N}_{[1,\overline{k}]}}T(dy'|y)+
		\underline{V}_{\adv,i}(\uns)\underbrace{\int_{\overline{X}\times\{0\}}T(dy'|y)}_0=&\\
		\underline{V}_{\adv,i}(\uns)\int_{\uns}T(dy'|y)
	\end{align*}
	where for crossing out the term $\int_{\overline{X}\times\{0\}}V_{i}(y')T(dy'|y)$ we used the fact that $\int_{\overline{X}\times\{0\}}T(dy'|y)=0$ (due to what is discussed in Remark \ref{rem:zero_tau}), and for the inequality we used Lemma \ref{lemma_uns}.
\end{proof}

Now, we are ready to prove Theorem \ref{main_theorem}:
\begin{proof}[\textbf{Proof of Theorem \ref{main_theorem}}]
	First, we prove the statement for cumulative rewards, and then we show how the proof is adapted for average and multiplicative rewards. We focus on the lower bound as the proof for the upper bound is similar. It suffices to show that there exists an adversary $\adv^*\in\Pi$ such that:
	\begin{equation}\label{eq:wannaprove}
		\expec^{q_0}_{\adv^*}[\underline{g}_{\cum,N}(\tilde{\omega})]\leq\expec_{\pr^{y_0}_\Y}[g_{\cum,N}(\omega)]
	\end{equation}
	
	By employing Lemma \ref{lemma:vi}, specifically equation \eqref{eq:lemma_vi3}, to prove \eqref{eq:wannaprove} it suffices to prove that
	there exists a $\adv^*\in\Pi$ such that for any $q_0\in Q_\imc\setminus\{\uns\}$ and any $y_0\in q_0$:
	\begin{equation}\label{eq:wannaprove2}
		\underline{V}_{\adv^*,0}(q_0)\leq V_{0}(y_0)
	\end{equation}
	Consider the following adversary for all $q\in Q_\imc$, $i\in\mathbb{N}_{[0,N-1]}$:
	\begin{equation}\label{eq:pistar}
		\adv^*(i,q,q') = \left\{\begin{aligned}
			&\int_{q'}T(dy'|y^*_{i}(q)), \text{ if }q\neq\uns\\
			&1, \text{ if }q=q'=\uns,\\
			&0, \text{ otherwise }
		\end{aligned}\right.
	\end{equation}
	where $y^*_i(q)=\argmin_{y\in q}V_i(y)$. Indeed $\adv^*\in\Pi$, since $\sum_{q'\in Q_\imc}\adv^*(i,q,q')=1$ for all $q\in Q_\imc$, and from \eqref{eq:probability_bounds} and \eqref{eq:kernel} it easily follows that $\check{P}(q,q')\leq\adv^*(i,q,q')\leq\hat{P}(q,q')$\footnote{Adopting the transition-kernel notation, it can be written that for $q\in Q_\imc\setminus\{\uns\}$, $\check{P}(q,q') \leq \min_{y\in q}\int_{q'}T(dy'|y)$, for any $q'\in Q_\imc$. Similarly for $\hat{P}$. Indeed it follows that $\check{P}(q,q')\leq\adv^*(i,q,q')\leq\hat{P}(q,q')$}.
	
	Now, we are ready to prove \eqref{eq:wannaprove2}, by induction. First, from \eqref{eq:underline_R} it is obvious that $\underline{V}_{\adv,N}(q_0)\leq V_{N}(y_0)$ for any $q_0\in Q_\imc\setminus\{\uns\}$ and any $y_0\in q_0$, since:
	\begin{align*}
		\underline{V}_{\adv,N}(q_0) = \underline{R}(q_0) \leq R(y_0) = V_N(y_0)
	\end{align*}
	Assume that $\underline{V}_{\adv,i}(q_0)\leq V_i(y_0)$ for any $q_0\in Q_\imc\setminus\{\uns\}$ and any $y_0\in q_0$, for some $i\in\mathbb{N}_{[1,N]}$. Then:
	\begin{equation*}
		\begin{aligned}
			\underline{V}_{\adv,i-1}(q_0) =&\\ \underline{R}(q_0) + \gamma\hspace{-2mm}\sum_{q'\in Q_\imc}\hspace{-3mm}\underline{V}_{\adv,i}(q')\adv(i-1, q_0,q')=&\\
			\underline{R}(q_0) + \gamma\sum_{q'\in Q_\imc\setminus\{\uns\}}\underline{V}_{\adv,i}(q')\int_{q'}T(dy'|y^*_{i-1}(q_0))+
			\gamma \underline{V}_{\adv,i}(\uns) \int_{\uns}T(dy'|y^*_{i-1}(q_0))\leq&\\
			\min_{y\in q_0}R(y) + \gamma\hspace{-5mm} \sum_{q'\in Q_\imc\setminus\{\uns\}}\min_{y\in q'}(V_i(y))\int_{q'}T(dy'|y^*_{i-1}(q_0))+
			\gamma \underline{V}_{\adv,i}(\uns) \int_{\uns}T(dy'|y^*_{i-1}(q_0))\leq&\\
			\min_{y\in q_0}R(y) + \gamma\hspace{-5mm} \sum_{q'\in Q_\imc\setminus\{\uns\}}\int_{q'}V_i(y')T(dy'|y^*_{i-1}(q_0))+
			\gamma \int_{\uns}V_i(y')T(dy'|y^*_{i-1}(q_0))\leq&\\
			R(y^*_{i-1}(q_0)) + \gamma\sum_{q'\in Q_\imc}\int_{q'}V_i(y')T(dy'|y^*_{i-1}(q_0))=&\\
			R(y^*_{i-1}(q_0)) + \gamma\int_{Q}V_i(y')T(dy'|y^*_{i-1}(q_0))=&\\
			V_{i-1}(y^*_{i-1}(q_0))=\min_{y\in q_0}V_{i-1}(y)
		\end{aligned}
	\end{equation*}
	where:
	\begin{itemize}
		\item in the first step we used \eqref{eq:viimc}; in the second step we used the definition \eqref{eq:pistar} of $\adv^*$;
		\item in the third step we used that $\underline{R}(q_0)=\min_{y\in q_0}R(y)$ (from \eqref{eq:underline_R}) and that $\underline{V}_{\adv,i}(q_0)\leq V_i(y_0)$ for any $q_0\in Q_\imc\setminus\{\uns\}$ and any $y_0\in q_0$ (from the induction assumption);
		\item in the fourth step we used that $\min_{y\in q'}(V_i(y))\leq V(y')$ for all $y'\in q'$, and the inequality given by Lemma \ref{lemma:uns-int};
		\item in the sixth step we used that $\bigcup_{q'\in q_\imc}q' = Q$, in the seventh step we used \eqref{eq:vi_system}, and in the last step we used that $y^*_{i-1}(q_0)=\argmin_{y\in q_0}V_{i-1}(y)$.
	\end{itemize}
	Hence, since $\underline{V}_{\adv,i-1}(q_0)\leq\min_{y\in _0}V_{i-1}(y)$, we have that \eqref{eq:wannaprove2} is proven by induction, thus proving \eqref{eq:wannaprove}.
	
	Only thing remaining is to explain how this proof generalizes to average and multiplicative rewards. The average reward is very simple, as it is just the time-average of a cumulative reward with $\gamma = 1$: $\expec[g_{\avg,N}(\omega)] = \tfrac{1}{N+1}\expec[\sum_{0}^NR(\omega(i))]$. Finally, for multiplicative rewards, only thing that changes w.r.t. cumulative rewards is the value iteration, which becomes:
	\begin{equation}\label{eq:viimc_mul}
		\begin{aligned}
			&\underline{V}_{\adv,N}(q) = \underline{R}(q)\\
			&\underline{V}_{\adv,i}(q) = \underline{R}(q) \cdot\sum_{q'\in Q_\imc}\hspace{-3mm}\underline{V}_{\adv,i+1}(q')\adv(i, q,q')
		\end{aligned}
	\end{equation}
\end{proof}

\subsection{Proofs of Statements from Section \ref{sec:transitions}}\label{sec:proofs_transitions_theorems}
\begin{proof}[\textbf{Proof of Lemma \ref{lemma:log-concave}}]
	This proof draws inspiration from the proof of \cite[Proposition 2]{Luca_hscc_2019}. Let us first prove log-concavity of the following simpler case:
	\begin{equation*}
		g(x) = \int_{S}\normal(dz|x,\Sigma)
	\end{equation*}
	with $S\subseteq\real^n$ not dependent on $x$. Observe that:
	\begin{equation*}
		g(x) = \int_{S-\{x\}}\normal(dz|0,\Sigma),
	\end{equation*}
	where $S-\{x\}$ is still a convex set as a mere translation of $S$. Then, $g(x)$ can be written as:
	\begin{equation*}
		g(x) = P(S-\{x\}),
	\end{equation*}
	where $P(\cdot)$ is a probability measure over $\borel(\real^n)$ induced by the distribution $\mathcal{N}(0,\Sigma)$. Since $\mathcal{N}(z|0,\Sigma)$ is log-concave, from \cite[Theorem 2]{prekopa1973logarithmic} we know that $P$ is a log-concave measure, meaning that for every pair of convex sets $S_1,S_2\subseteq\real^n$ and any $\lambda\in(0,1)$:
	\begin{equation}\label{eq:logconcmeasure}
		P(\lambda S_1 + (1-\lambda)S_2)\geq (P(S_1))^\lambda(P(S_2))^{1-\lambda}
	\end{equation}
	Moreover, for any $x_1,x_2\in\real^n$ and any $\lambda\in(0,1)$ we have:
	\begin{equation}\label{eq:logconcmeasure2}
		\begin{aligned}
			\lambda(S-\{x_1\})+(1-\lambda)(S-\{x_2\})=&\\ \{\lambda(y- x_1):y\in S\} + \{(1-\lambda)(w- x_1):w\in S\}=&\\
			\{\lambda(y- x_1) + (1-\lambda)(w- x_1): y,w\in S\}=&\\
			\{\lambda y + (1-\lambda)w -\lambda x_1 - (1-\lambda)x_2:y,w\in S \}=&\\
			\{\lambda y + (1-\lambda)w: y,w\in S\}-\lambda\{x_1\}-(1-\lambda)\{x_2\} =&\\
			S -\lambda\{x_1\}-(1-\lambda)\{x_2\}\quad&
		\end{aligned}
	\end{equation}
	where the last equality is because $v=\lambda y + (1-\lambda)w$ is a convex combination of any two points $y,w\in S$ and $S$ is convex\footnote{Since $S$ is convex, then for any two $y,w\in S$ and any $\lambda\in(0,1)$ we have that $v=\lambda y + (1-\lambda)w\in S$. Thus, for a given $\lambda,$ $\{\lambda y + (1-\lambda)w: y,w\in S\}\subseteq S$. But, also, $S = \{\lambda y + (1-\lambda)y\in S: y\in S\} \subseteq \{\lambda y + (1-\lambda)w: y,w\in S\}$. Thus, it has to be $S = \{\lambda y + (1-\lambda)w: y,w\in S\}$.}. 
	
	Finally, for any $x_1,x_2\in\real^n$ and any $\lambda\in(0,1)$ we have:
	\begin{align*}
		g(\lambda x_1 + (1-\lambda)x_2)=&\\
		P\Big(S - \lambda \{x_1\} + (1-\lambda)\{x_2\}\Big)=&\\
		P\Big(\lambda(S-\{x_1\})+(1-\lambda)(S-\{x_2\})\Big)\geq&\\
		\Big(P(S-\{x_1\})\Big)^\lambda \Big(P(S-\{x_2\})\Big)^{(1-\lambda)}=&\\
		(g(x_1))^\lambda(g(x_2))^{1-\lambda}\quad&
	\end{align*}
	where in the second equality we used \eqref{eq:logconcmeasure2} and for the inequality we used \eqref{eq:logconcmeasure}. Thus, it follows that $g(x)$ is log-concave. 
	
	For the general case, since $S(x)\subseteq \real^m$ is linear on $x$ and convex, then it can be written as $S(x) = S' + \{Gx\}$, where $S'\subseteq\real^m$ is convex and $G\in\real^{m\times n}$. Thus we have:
	\begin{align*}
		h(x) = \int_{S(x)}\normal(dz|f(x),\Sigma) &= \int_{S' + \{Gx\}}\normal(dz|f(x),\Sigma) \\&= \int_{S'}\normal(dz|f(x)-Gx,\Sigma)\\&=g(f(x)-Gx)
	\end{align*}
	where $g(x) = \int_{S'}\normal(dz|x,\Sigma)$. The function $h(x) = g(f(x)-Gx)$ is log-concave as the composition of the log-concave function $g(x)$ with the affine function $f(x)-Gx$.
\end{proof}
\bibliography{stoch_abs_bib.bib} 
\bibliographystyle{IEEEtran} 

\end{document}